\documentclass[11pt]{article}
\usepackage{latexsym}
\usepackage{amssymb,amsmath}
\usepackage{amsthm}
\usepackage{mathtools}
\usepackage{graphicx}
\usepackage{color}
\usepackage{blkarray}
\usepackage{multirow}
\usepackage{hyperref}
\usepackage{float}
\usepackage{subfig}
\usepackage{relsize}

\usepackage[numbers,sort&compress]{natbib}

\hypersetup{
  colorlinks   = true, 
  urlcolor     = blue, 
  linkcolor    = blue, 
  citecolor   = red 
}

\usepackage{breqn}

\usepackage{xcolor}
\usepackage{cancel}

\usepackage{hyperref}
\usepackage{empheq}

\usepackage[utf8]{inputenc}

\usepackage{setspace}

\usepackage{enumerate}

\usepackage[titletoc,toc]{appendix}

\usepackage[font=small,labelfont=bf]{caption}

\newcommand{\ds}{\displaystyle}
\newcommand{\mc}{\mathcal}

\DeclareMathOperator{\argmax}{argmax}
\DeclareMathOperator{\supp}{supp}

\DeclareMathOperator{\argsup}{argsup}

\newcommand{\bbm}{\begin{bmatrix}}
\newcommand{\bpm}{\begin{pmatrix}}
\newcommand{\ebm}{\end{bmatrix}}
\newcommand{\epm}{\end{pmatrix}}

 \newcommand{\del}[2]{\frac{\partial #1}{\partial #2}}
 \newcommand{\dsdel}[2]{\displaystyle\frac{\partial #1}{\partial #2}}

\newcommand{\dsddx}[2]{\displaystyle\frac{\mathrm{d} #1}{\mathrm{d} #2}}
\newcommand{\dsddt}[1]{\displaystyle\frac{\mathrm{d} #1}{\mathrm{d}t}}

\newcommand{\dsint}{\mathrm{d}s}

\newcommand{\dx}{\mathrm{d}x}
\newcommand{\dy}{\mathrm{d}y}

\newcommand{\holder}{Hölder  \hspace{0.05mm}}

\renewcommand{\abstractname}{Abstract}
\numberwithin{equation}{section}

\usepackage[T1]{fontenc}
\usepackage[utf8]{inputenc}
\usepackage{authblk}

\title{\large{The Replicator Dynamics for Multilevel Selection in Evolutionary Games}}
\author[1]{Daniel B. Cooney}
\affil[1]{Program in Applied and Computational Mathematics, Princeton University}

\usepackage{minitoc}

\newcommand{\myindent}{\hspace{10mm}}

\begin{document}

\renewcommand{\baselinestretch}{1.1}
\newtheorem{definition}{Definition}[section]
\newtheorem{theorem}{Theorem}[section]
\newtheorem{lemma}[theorem]{Lemma}
\newtheorem{corollary}[theorem]{Corollary}
\newtheorem{claim}[theorem]{Claim}
\newtheorem{fact}[theorem]{Fact}
\newtheorem{proposition}{Proposition}[section]
\newtheorem{remark}{Remark}[section]
\newtheorem{example}{Example}[section]

\newcommand{\qedsymb}{\mbox{ }~\hfill~{\rule{2mm}{2mm}}}
\newenvironment{proof1}{\begin{trivlist}
\item[\hspace{\labelsep}{\bf\noindent Proof: }]
}{\qedsymb\end{trivlist}}

\newenvironment{hackyproof}{\begin{trivlist}
\item[\hspace{\labelsep}{\bf\noindent Proof: }]
}{
\end{trivlist}}

\maketitle

\begin{abstract}
\begin{spacing}{0.1}
We consider a stochastic model for evolution of group-structured populations in which interactions between group members correspond to the Prisoner's Dilemma or the Hawk-Dove game. Selection operates at two organization levels: individuals compete with peer group members based on individual payoff, while groups also compete with other groups based on average payoff of group members. 
In the Prisoner's Dilemma, this creates a tension between the two levels of selection, as defectors are favored at the individual level, whereas groups with at least some cooperators outperform groups of defectors at the between-group level.  In the limit of infinite group size and infinite number of groups, we derive a non-local PDE that describes the probability distribution of group compositions in the population. For special families of payoff matrices, we characterize the long-time behavior of solutions of our equation, finding a threshold intensity of between-group selection required to sustain density steady-states and the survival of cooperation.
When all-cooperator groups are most fit, the average and most abundant group compositions at steady-state range from featuring all-defector groups when individual-level selection dominates to featuring all-cooperator groups when group-level selection dominates. When the most fit groups have a mix of cooperators and defectors, then the average and most abundant group compositions always feature a smaller fraction of cooperators than required for the optimal mix, even in the limit where group-level selection is infinitely stronger than individual-level selection.   
In such cases, the conflict between the two levels of selection cannot be decoupled, and cooperation cannot be sustained at all in the case where between-group competition favors an even mix of cooperators and defectors.
\end{spacing}
\end{abstract}

\singlespacing
{\hypersetup{linkbordercolor=black, linkcolor = black}
\begin{spacing}{0.01}
\renewcommand{\baselinestretch}{0.1}\normalsize
\tableofcontents
\addtocontents{toc}{\protect\setcounter{tocdepth}{2}}
\end{spacing}

\section{Introduction}
Across many complex biological systems, ecological and evolutionary dynamics can best be understood by considering forces of competition and selection acting at multiple organizational levels. Competition at multiple levels is a driving force for the evolution of virulence, where pathogen strains which replicate quickly within-host can outcompete rival strains, but are also more likely to kill the host, preventing further transmission of the strain \citep{gilchrist2004optimizing,gilchrist2006evolution,coombs2007evaluating}. The balance of these two dynamics has been theoretically demonstrated to select for the evolution of strains with intermediate virulence, and has been posited as an explanation for the diminished virulence of myxomatosis in rabbit populations in Australia \citep{levin1981selection,dwyer1990simulation}. %

\myindent Multilevel selection is also considered to be a factor in prebiotic evolution and the evolution of multicellularity \citep{szathmary1987group,hogeweg2003multilevel, takeuchi2009multilevel, takeuchi2012evolutionary, tarnita2013evolutionary,mathis2017prebiotic}. In an attempt to answer the question ``Which came first, the protein or the nucleic acid?''\citep{eigen1971selforganization}, Eigen used the quasispecies equation to show that competition of self-replicating genes leads to the ultimate dominance of a single type, implying the inability of quasispecies competition to generate genomes of greater biological complexity \citep{eigen1971selforganization,szathmary1995major}. To overcome this hurdle, Szathmary and Demeter introduced the ``stochastic corrector'' model, which places component elements into separated vesicles which are themselves capable of self-replication \citep{szathmary1987group}. Considering two constituent self-replicating quasispecies, Szathmary and Demeter show that the coupling of selection within vesicles and between vesicles allow for the coexistence of competing quasispecies that cannot coexist in the absence of vesicle structure \citep{szathmary1987group}. %

\myindent These hierarchical biological phenomena can often be viewed in the framework of game theory, where a population of individuals is subdivided into groups, and the strategic behavior of individuals has effects both on competition for resources or reproduction between individuals in the same group and in higher-level competition where groups compete with other groups based on their respective strategic compositions. These ideas have been explored in simulation studies for the evolution of protocells and the origins of life \citep{markvoort2014computer,hogeweg2003multilevel}. Similar phenomena have been observed in field and laboratory research on the ant species {\it Pogonomyrmex californicus} \citep{shaffer2016foundress}, whose colonies can be founded by multiple unrelated queens, and where characteristics such as density of ant hills can serve as a mechanism for mediating the individual advantage for aggressive, solitary ant queens and the greater capacity for cooperative queens to establish shared ant hills. 

\myindent Models of multilevel selection have been studied in the population genetics literature, starting with the Kimura equation with deme-level selection events \citep{kimura1955solution}, which uses a diffusion expansion to predict the fixation probabilities of two alleles, one which dominates in within-group competition, while the other type dominates in between-group selection. Ogura and Shimakura then studied the Kimura equation, characterizing invariant families of solutions and demonstrating convergence to steady-state distributions of alleles \citep{ogura1987stationary,ogura1987stationary2}. Fontanari and Serva further generalize the Kimura model to allow for nonlinear group-level reproduction functions to describe the evolution of protocells, and show that group reproduction functions that favor coexistence of an even mix of cooperators and defectors can result in steady-states with groups of coexisting cooperators and defectors \citep{fontanari2014effect,fontanari2014nonlinear,fontanari2013solvable}.

\myindent In the context of evolutionary games, Traulsen and Nowak introduced a model of multilevel selection in finite populations, in which within-group selection follows the frequency-dependent Moran process and groups can also undergo fission when they reach a maximum size \citep{traulsen2006evolution,traulsen2008analytical}. A similar approach was followed by Böttcher and Nagler, who explored the frequency-independent donation game and added the possibility of group-level extinction events that depend on average group payoff \citep{bottcher2016promotion}. In both models, the authors derive conditions to favor the fixation of a single cooperator in a resident population of defectors relative to the fixation of a single defector in a resident population of cooperators. For games in which within-group or between-group selection favors coexistence of cooperation and defection, it becomes important to study the dynamics for the full distribution of group compositions, as fixation is not necessarily the predicted long-time behavior.

\myindent In this spirit, Luo introduced a stochastic ball-and-urn process model for multilevel selection with birth and death events for both individuals and groups \citep{luo2014unifying}. In this model, the population is subdivided into groups; and agents can be of one of two types: those with an advantage at the within-group level of selection (type-$I$) and those who contribute to success at the between-group level of selection (type-$G$). Selection in this model is frequency-independent, as type-$G$ individuals reproduce at rate $1$ and type-$I$ individuals reproduce at faster rate $1+s$, regardless of the composition of their group. Groups reproduce at rates proportional to their fraction of type-$G$ individuals, so the two levels of selection act in direct opposition.
In the limit of a large number of groups and large group size, Luo obtains deterministic population dynamics governed by a non-local (or integro-) PDE 
\begin{equation} \label{eq:luomattingly}
\dsdel{f(t,x)}{t} = \dsdel{}{x}\left[x(1-x) f(t,x) \right] + \lambda f(t,x) \left[x - \int_0^1 y f(t,y) \dy \right]
\end{equation}
where $f(t,x)$ is the probability density for a group with a fraction $x$ of type-$I$ individuals at time $t$ and $\lambda$ governs the relative intensities of selection at the two levels.
The Luo model was extended by van Veelen et al to explore this multilevel selection model in the context of the debate between kin selection and group selection \citep{van2014simple}. Luo and Mattingly rigorously proved the weak convergence of the ball-and-urn process to a measure-valued PDE in the deterministic limit as $M,N \to \infty$ 
and characterize the long-time behavior of Equation \ref{eq:luomattingly} based on the initial distribution of groups in the population and the relative selection strength $\lambda$ \citep{luo2017scaling}. McLoone et al also study a similar two-level Moran process for a special case of the Hawk-Dove game, showing in simulations that multilevel selection produces higher levels of cooperation than the traditional well-mixed model \citep{mcloone2018stochasticity}. Similar models were introduced by Simon et al, which include more heterogeneity in group size and a broader variety of group-level events such as group fission, fusion, and extinction \citep{simon2010dynamical,simon2012numerical,simon2013towards,simon2016group,puhalskii2017large}. 
Numerical analysis has shown that group fission and extinction effects can facilitate the emergence of cooperation \citep{simon2016group}.

\myindent Notably, the Luo model was, in part, designed as a minimal model illustrating the direct competition between two-forces of selection: within-group selection that favors defectors, and between-group selection which favors groups with many cooperators \citep{luo2014unifying}. In evolutionary games, one or both of individual payoff and average group payoff can either facilitate the dominance of one type or promote coexistence. At the individual level, Prisoner's Dilemma games promote the dominance of defection, whereas Hawk-Dove games promote coexistence between cooperators and defectors. At the between group level, both Prisoner's Dilemmas and Hawk-Dove games can either exhibit maximal group fitness for all cooperator groups or provide maximal benefit to groups with an intermediate fraction of cooperators and defectors. An interesting question to ask is whether generating individual and group level selection from payoffs of an evolutionary game produces different qualitative behaviors than the frequency-independent Luo model.

\myindent The question of achieving maximal average payoff in a group has been recently studied in theoretical and empirical contexts. One mechanism by which group payoff can decrease for high levels of cooperation is subadditivity of payoffs, a characteristic feature of nonlinear public goods game, in which the amount of public good produced can either be a step-function  \citep{pacheco2009evolutionary,souza2009evolution} or a smoothed Fermi function \citep{archetti2011coexistence} of the fraction of cooperators. Maclean et al explore experimentally the population fitness in yeast groups composed of individuals of a ``cooperator'' type, which pays a metabolic cost to produce a protein invertase which decomposes complex carbohydrates, and a ``cheater'' type which consumes decomposed sugars but freerides on the efforts of the cooperators. They show that, under certain conditions, population fitness is maximized by groups which have a combination of cooperators and cheaters  \citep{maclean2010mixture}. Boza and Sz{\'a}mad{\'o} also consider optimal provision of public goods in a simulation of multilevel selection in animal groups. They show that such selection favors convergence to the optimal fraction of cooperators, and refer to the defectors that help to achieve this optimal cooperation as ``beneficial laggards'' \citep{boza2010beneficial}.

\myindent Here, we extend Luo's ball-and-urn model to allow selection at both the within-group and between-group levels to depend on payoffs from a two-player cooperative dilemma. By formulating between-group selection as dependent on the average payoff of group members, we allow for the misalignment of individual and group interests to emerge from the payoff matrices of the Prisoner's Dilemma and the Hawk-Dove game. We then derive a limiting PDE description of this multilevel selection process, which can be understood as coupling the replicator dynamics for within-group selection and a ``group-level replicator dynamics'' that characterizes birth and death events at the group level. We then study the steady-states of our multilevel PDE for both the Prisoner's Dilemma and Hawk-Dove game, and we determine the most abundant group composition at steady-state.

\myindent We characterize the relative levels of between-group and within-group selection strength at which cooperation can coexist with defection, and show that whether the limit of weak within-group selection produces steady-states whose most abundant groups have the highest average payoff depends on whether group average payoff is maximized at full cooperation or if group average payoff achieves its optimum with an intermediate fraction of cooperators. 
In Section \ref{sec:Derivation}, we define our evolutionary game and two-level Moran process and use them to derive a limiting PDE. In Section \ref{sec:PD}, we focus on the Prisoner's Dilemma, proving long-time behavior for a special family of payoff matrices, and then studying the steady-states for more general PD games. In Section \ref{sec:HD}, we perform a similar analysis for the Hawk-Dove game. We postpone some detailed calculations and present them in the Appendix. 

\section{Derivation of Differential Equations} \label{sec:Derivation}

In this paper, we extend the framework of Luo and Mattingly to discuss multilevel selection where the underlying individual and group-level reproductive fitness is generated by two-player evolutionary games. We consider games with two strategies: cooperation ($C$) and defection ($D$).  Our population consists of $M$ groups with $N$ individuals per group. Individuals obtain reproductive fitness by playing the following symmetric two-strategy game against the  other $N-1$ members of its group
\begin{equation} \label{eq:generalpayoffmatrix}
\begin{blockarray}{ccc}
& C & D \\
\begin{block}{c(cc)}
C & R & S \\
D & T & P \\
\end{block}
\end{blockarray}
\end{equation}
where  $R$ is the reward for mutual cooperation, $T$ is the temptation to defect against a cooperator, $P$ is the punishment for mutual defection, and $S$ is the sucker payoff for cooperating with a defector. We define a Prisoner's Dilemma as a game with the payoff rankings $T > R > P > S$, whereas the Hawk-Dove (or Snowdrift) game is defined by $T > R > S > P$ \citep{nowak2006five}. Because payoff will determine birth rates for individuals and groups, we assume for convenience that $R,S,T,P \geq 0$. 

In a group with $i$ cooperators, a cooperator and defector receive average payoffs from interactions of
\begin{align}
F_i^C &= \frac{1}{N-1} \left[ (i-1)R + (N - i) S \right] \\
F_i^D &= \frac{1}{N-1} \left[ i T + (N-i-1) P \right]
\end{align}
Then the average payoff of all members of an $i$-cooperator group is 
\begin{align} G_i &= \frac{ i F_i^C + (N-i) F_i^D}{N} 
= \frac{i \left[S+T-2P\right]}{N-1} + \frac{i^2 \left[ R - S - T + P \right]}{N(N-1)} + \frac{N^2 P}{N(N-1)} + \frac{i(P-R) - NP}{N(N-1)}  \end{align} 
Denoting $x = \frac{i}{N}$, we see that $ G_{x = \frac{i}{N}} \to G(x) = (S+T-2P) x + (R-S-T+P)x^2 + P$ as $N \to \infty$.
Within-group dynamics follow a continuous-time Moran process. Individuals of type $X$ in a group with $i$ cooperators are chosen to give birth at rate $1 + w_I F_i^X$, where $w_I$ is the strength of selection at the within-group or individual level. The offspring replace a randomly chosen individual in the same group, including possibly its parent. 

Denoting by $f_i(t)$ the fraction of groups with $i$ cooperators, we see that $f_i(t)$ increases by $\frac{1}{M}$ due to within-group competition if one of two events happens
\begin{itemize}
\item A new cooperator replaces a defector in a group with $i-1$ cooperators, which happens with rate
$ M f_{i-1}(t) (i-1)\left(1 + w F_{i-1}^C\right) \left(1 - \frac{i-1}{N}\right)  $
\item A new defector replaces a cooperator in a group with $i+1$ cooperators, which happens with rate
$M f_{i+1}(t) (N - (i+1)) \left(1 + w F_{i+1}^D \right) \left(\frac{i+1}{N} \right) $

\end{itemize}

Similarly, $f_i(t)$ decreases by $\frac{1}{M}$ due to within-group competition when 

\begin{itemize}

\item A new cooperator replaces a defector in a group with $i$ cooperators, which happens with rate
$ M f_{i}(t) (i)\left(1 + w F_{i}^C\right) \left(1 - \frac{i}{N}\right)  $
\item A new defector replaces a cooperator in a group with $i$ cooperators, which happens with rate
$M f_{i}(t) (N - i) \left(1 + w F_{i}^D \right) \left(\frac{i}{N} \right) $

\end{itemize}
and within-group events in which a cooperator replaces a cooperator or a defector replaces a defector leave $f_i(t)$ unchanged.

The dynamics of between-group competition follow a process analogous to a continuous-type Moran process. A group with $i$ cooperators is chosen to produce a copy of itself at rate $\Lambda \left(1 + w_G G_i\right)$, where $w_G$ represents the selection strength at the between-group level and $\Lambda$ modulates the relative rate of within-group and between-group replication events. The offspring group replaces a randomly chosen group in the population.

 The fraction of $i$ cooperator groups $f_i(t)$ increases by $\frac{1}{M}$ in the between-group dynamics when a group of $i$ cooperators is selected to reproduce and a group with a different number of cooperators is selected to die. This occurs with rate  $\Lambda M f_i(t) \left( 1 + w G_i \right)(1- f_i(t))$. 
The between-group dynamics result in a decrease  of $\frac{1}{M}$ for $f_i(t)$ when a group with $i$ cooperations is chosen to die and a group of different size is selected to replace it, which occurs with rate $\Lambda M f_i(t) \left(\sum_{j \ne i} f_j(t) \left( 1 + w G_j \right) \right)$.

Now we follow the heuristic derivation of Luo \citep{luo2014unifying} and of Van Veelen et al \citep{van2014simple}, obtaining limiting ODE and PDE descriptions of our multilevel system by first taking the limit as the number of groups $M \to \infty$, and then taking the limit as group size $N \to \infty$. As an aside, we note that first taking the limit of group size $N \to \infty$ would yield a system of $M$ ODEs for evolutionary game dynamics in $M$ infinitely large subpopulations, such as studied by Young and Belmonte \citep{young2018fast}.  Luo and Mattingly also show that the multilevel process weakly converges to a deterministic measure-valued process and, in a different scaling limit, can weakly converge to a martingale-valued process known as a Fleming-Viot process\citep{luo2017scaling,dawson2013multilevel}
. A similar rigorous derivation is given by Puhalskii et al for for Simon's multilevel model with group-level fission and extinction events \citep{puhalskii2017large}.

The infinitesimal mean for this continuous-time stochastic process is given by 
\begin{align*} \hspace{-5mm} M_{\Delta t} :=  \mathrm{E}\left[f_i(t+\Delta t) - f_i(t) \right] &= \frac{1}{M} \left[ P\left(f_i(t + \Delta t) - f_i(t) = \frac{1}{M}\right) -  P\left(f_i(t + \Delta t) - f_i(t) = \frac{-1}{M}\right) \right]   \\ &= \frac{1}{M} \left( M f_{i-1}(t) (i-1)\left(1 + w_I F_{i-1}^C\right) \left(1 - \frac{i-1}{N}\right) \right) \Delta t  \\ &+ \frac{1}{M} \left(M f_{i+1}(t) (N - (i+1)) \left(1 + w_I F_{i+1}^D \right) \left(\frac{i+1}{N} \right)\right) \Delta t \\ &- \frac{1}{M} \left(
M f_{i}(t) (i)\left(1 + w_I F_{i}^C\right) \left(1 - \frac{i}{N}\right) \right) \Delta t \\ &- \frac{1}{M} \left( M f_{i+1}(t) (N - (i+1)) \left(1 + w_I F_{i+1}^D \right) \left(\frac{i+1}{N} \right) \right) \Delta t \\ &+ \frac{1}{M} \left( \Lambda M f_i(t) \left( 1 + w_G G_i \right)(1- f_i(t)) - \Lambda M f_i(t) \left(\ds\sum_{j \ne i} f_j(t) \left( 1 + w_G G_j \right) \right)  \right) \Delta t \end{align*} 
We can rearrange this to say that 
\begin{align*} \frac{ \mathrm{E}\left[f_i(t+\Delta t) - f_i(t) \right]}{\Delta t} &= \frac{1}{N} D_2\left(f_i(t) \frac{i}{N} \left(1 - \frac{i}{N}\right) \right) + \Lambda w_G f_i(t)  \left( G_i - \ds\sum_{j=0}^N G_j f_j(t) \right)  \\ &+ w_I \left[D_1^+ \left(f_i(t) \frac{i}{N}(1 - \frac{i}{N}) F_i^D \right) - D_1^- \left(f_i(t) \frac{i}{N} (1- \frac{i}{N}) F_i^C \right)  \right] \end{align*}
where $D_2(\cdot)$, $D_1^+(\cdot)$ and $D_1^-(\cdot)$ are second-order, first-order forward, and first-order backward difference quotients given by the formulas
\[D_1^+\left(u\left(\tfrac{i}{N}\right) \right) := \frac{u(\frac{i+1}{N}) - u(\frac{i}{N})}{\frac{1}{N}} \: \:, \: \: D_1^-\left(u(\tfrac{i}{N})\right) := \frac{u(\frac{i}{N}) - u(\frac{i-1}{N})}{\frac{1}{N}} \] \[D_2\left(u(\tfrac{i}{N}) \right) := \frac{u(\frac{i+1}{N}) - 2 u(\frac{i}{N}) + u(\frac{i-1}{N})}{\frac{1}{N^2}} \] 

Dividing both sides by $w_I$, taking the limit $\Delta t \to 0$, and rescaling time as $\tau = \frac{t}{w_I}$, we have \begin{align*} \dsddx{E[f_i(t)]}{\tau} &= \frac{1}{N} D_2\left(f_i(t) \frac{i}{N} (1 - \frac{i}{N}) \right) + \lambda f_i(t)  \left( G_i - \ds\sum_{j=0}^N G_j f_j(t) \right)  \\ &+  D_1^+ \left(f_i(t) \frac{i}{N}(1 - \frac{i}{N}) F_i^D \right) - D_1^- \left(f_i(t) \frac{i}{N} (1- \frac{i}{N}) F_i^C \right) \end{align*}
where we have introduced the parameter $\lambda := \frac{\Lambda w_G}{w_I}$ to jointly measure the relative rate of between-group and within-group selection events ($\Lambda$) and the relative strengths of within-group ($w_I$) and between-group ($w_G$) selection.
By similar calculations, we see that the infinitesimal variance of $f_i(t)$ is %
\begin{align*} V_{\Delta t} &:= E \left[ \left(f_i(t + \Delta t) - f_i(t) \right)^2   \right]  = \frac{1}{M^2} P\left(f_i(t + \Delta t) - f_i(t) = \frac{1}{M}\right) + \frac{1}{M^2}  P\left(f_i(t + \Delta t) - f_i(t) = \frac{-1}{M}\right) \\ &= \frac{1}{M^2} \left[ O(M) \right] \to 0 \: \: \mathrm{as} \: \: M \to \infty \end{align*} As demonstrated by Luo \citep{luo2014unifying}, the vanishing infinitesimal variance in the large $M$ limit tells us that $f_i(t)$ evolves deterministically and takes on the constant value %
$E[f_i(t)]$. In this limit, our multilevel system takes on the deterministic description given by the system of ODEs 
\begin{align*} \dsddt{f_i(t)} &= \frac{1}{N} D_2\left(f_i(t) \frac{i}{N} (1 - \frac{i}{N}) \right) + \lambda f_i(t)  \left( G_i - \ds\sum_{j=0}^N G_j f_j(t) \right) \\ &+ \left[D_1^+ \left(f_i(t) \frac{i}{N}(1 - \frac{i}{N}) F_i^D \right) - D_1^- \left(f_i(t) \frac{i}{N} (1- \frac{i}{N}) F_i^C \right)  \right] \end{align*}
From the ODE limit, we take the limit as group membership $N \to \infty$ to obtain the limiting PDE \begin{align}  \label{eq:generalPDE} \dsdel{f(t,x)}{t} &= \dsdel{}{x}\left[ x (1-x) \left((P-S) - (R - S -T + P) x \right) f(t,x) \right] \\ &+ \lambda f(t,x) \left[ (S+T - 2P) \left( x - \ds\int_0^1 y f(t,y) \dy \right) + (R-S-T+P) \left( x^2 - \ds\int_0^1 y^2 f(t,y) \dy \right) \right] \nonumber  \end{align}
where $f(t,x)$ denotes the probability density of groups with composition of fraction $x$ cooperators and $1-x$ defectors. We note that the term with 2nd-order difference quotient $ \frac{1}{N} D_2\left(f_i(t) \frac{i}{N} (1 - \frac{i}{N}) \right)$ vanishes when $N \to \infty$, so their is no diffusive effect in the PDE limit as their is in the ODE limit with infinite $M$ and finite $N$. We note that this is a consequence of the scaling of $M$ and $N$, which allows us to compare deterministic effects of within-group and between-group competition, although others have explored multilevel selection in limits that preserve a diffusion term \citep{ogura1987stationary,ogura1987stationary2,fontanari2014effect,fontanari2014nonlinear,fontanari2013solvable} or in the stochastic Fleming-Viot limit \citep{luo2017scaling,dawson2013multilevel,blancas2017representation}.

\subsection{Interpretation of PDE}

Going forward, we will write this more compactly by denoting $\alpha = R - S - T + P$, $\beta = S-P$, $\gamma = S + T - 2P$, and $M_j^f = \int_0^1 y^j f(t,y) \dy$, the moments of $f(t,x)$, to give us \begin{equation} \label{eq:pdeparam} \dsdel{f(t,x)}{t} = - \dsdel{}{x} \left[ x(1-x) (\beta + \alpha x ) f(t,x) \right] + \lambda f(t,x) \left[ \gamma x + \alpha x^2 - \left( \gamma M_1^f + \alpha M_2^f \right) \right] \end{equation}
which corresponds to the dynamics of our multilevel system with interactions corresponding to a payoff matrix of the form 
\begin{equation} \label{eq:parameterpayoffmatrix}
\begin{blockarray}{ccc}
& C & D \\
\begin{block}{c(cc)}
C &  \gamma + \alpha + P &   - \beta + P \\
D &  \gamma - \beta + P & P \\
\end{block}
\end{blockarray}
\end{equation}
We choose to retain the punishment payoff $P$ from the original payoff matrix of Equation \ref{eq:generalpayoffmatrix} and to express the payoffs $R$, $S$, and $T$ in terms of $P$, $\alpha$, $\beta$, and $\gamma$ because $P$ directly shows up in the function for average group payoff $G(x) = P + \gamma x + \alpha x^2$. Under suitable rescaling of time, we can normalize $P$ to an aribitrary non-negative value. Further, because the multilevel dynamics given by Equation \ref{eq:pdeparam} depend on the payoff matrix through $\alpha$, $\beta$, and $\gamma$, we can understand the role of the payoff matrix through the relative values of these three parameters.

The term $\alpha = R - S - T + P$ is called the ``gains from switching'' and determines whether the average group payoff $G(x)$ is concave or convex, and can be positive, negative, or zero for the Prisoner's Dilemma and is always negative for the Hawk-Dove game. The term $\beta = S - P$ characterizes whether it is worse to be a cooperator or defector when interacting with a defector, characterizing the main difference between the interactions of the Hawk-Dove game ($\beta > 0$) and the Prisoner's Dilemma ($\beta < 0$).  The parameter $\gamma = S + T - 2P$ tells us whether the total payoff for an interaction between a cooperator and a defector exceeds the total payoff of an interaction between two defectors. It is always positive for the Hawk-Dove game, and can be negative, positive, or zero for the Prisoner's Dilemma. 

\begin{example} Consider a frequency-dependent Prisoners' Dilemma with the following payoff matrix 
\begin{equation} \label{eq:bcpayoffmatrix}
\begin{blockarray}{ccc}
& C & D \\
\begin{block}{c(cc)}
C &  b - c &   -c \\
D &  b & 0 \\
\end{block}
\end{blockarray}
\end{equation}
which has the interpretation that cooperators pay a cost $c$ to confer a benefit $b$ to their coplayer, while defectors pay no cost and confer no benefit to their coplayer. This game is also called the donation game and is a two-player version of a linear public goods game \citep{nowak} and was mentioned as an example application for Luo's multilevel framework \citep{luo2014unifying}. For this game, the multilevel dynamics are given by the equation \begin{equation} \label{eq:pdfreqind} \ds\del{f(t,x)}{t} = c \dsdel{}{x}\left[ x(1-x) f(t,x) \right] + \lambda (b-c) f(t,x) \left[ x - M_1^f \right]  \end{equation} which is a rescaled version of Luo and Mattingly's multilevel equation \citep{luo2014unifying,luo2017scaling}. \end{example}
\begin{example}
A standard class of Hawk-Doves games, dating back to the ideas of Maynard Smith and Price, is characterized by the family of payoff matrices of the form 
\begin{equation} \label{eq:VCpayoffmatrix}
\begin{blockarray}{ccc}
& C & D \\
\begin{block}{c(cc)}
C &  \ds\frac{V}{2} &   0 \\
D &  V & \ds\frac{V-C}{2} \\
\end{block}
\end{blockarray}
\end{equation}
\citep{smith1973logic,smith1982evolution, nowak}. Here, a resource of value $V$ is to be divided beteween two individuals who can either cooperate (``doves'') or defect (``hawks''). A pair of doves split the resource in half and each receive $\frac{V}{2}$, while an interaction between a hawk and a dove results in the full resource $V$ for the hawk and nothing for the dove. When a pair of hawks meet, they fight over the resource, and in expectation receive half of the resource but also pay a cost $\frac{C}{2}$ to fight, obtaining a total payoff of $\frac{V-C}{2}$. It is assumed that $C > V$, so that the expected outcome for a fight between hawks is worse that the outcome for a dove that allows a hawk to take the full resource. For this family of payoff matrices, the multilevel dynamics are given by the equation 
\begin{equation} \frac{2}{C} \label{eq:HDCVpde} \ds\del{f(t,x)}{t} =  \dsdel{}{x}\left[ x(1-x)\left( \frac{C-V}{C} - x\right) f(t,x) \right] + \lambda f(t,x) \left[ 2x - x^2  -\left( 2 M_1^f - 2 M_2^f \right) \right]  \end{equation}
where within-group dynamics push for coexistence of $\frac{C-V}{C}$ cooperators and $\frac{V}{C}$ defectors and where between-group dynamics are most favorable to all-cooperator groups.
\end{example}
For general two-player, two-strategy games, the PDE for our multilevel system given by Equation \ref{eq:pdeparam} is a first-order equation with characteristic curves given by the replicator dynamics \begin{equation} \label{eq:replicator} \dot{x}(t) = x (1-x) \left( \beta + \alpha x \right) \end{equation} For the Prisoner's Dilemma, the replicator dynamics has a stable equilibrium at $0$ and an unstable equilibrium at 1, while for the Hawk-Dove game, there are unstable equilibria at 0 and 1 and a stable interior equilibrium at \[x^{eq}_{int} = - \ds\frac{\beta}{\alpha} = \ds\frac{S - P}{S - P + T - R} \in(0,1)\]. 

The advection term $ - \left[ x(1-x) (\beta + \alpha x ) f(t,x) \right]_x$ corresponds to the within-group population dynamics due to births and deaths of individuals. In the absence of group selection ($\lambda = 0$), equation \ref{eq:pdeparam} reduces to $f_t(t,x) =  \left[x(1-x) (\beta + \alpha x) f(t,x) \right]_x$, which can be seen as an Eulerian version of the replicator dynamics for a metapopulation without between-group interactions \citep{chalub2014frequency}. In this case, given density initial data $f(0,x) = f_0(x)$, $f(t,x)$ converges to a delta-distribution at the stable equilibrium at $t \to \infty$.

\par The nonlocal term $\lambda f(t,x) [(\gamma x + \alpha x^2) - (\gamma M_1^f + \alpha M_2^f )]$ describes the %
dynamics %
due to group-level extinction and replication events,
and can be interpreted as a version of the replicator dynamics for between-group competition. In the absence of within-group selection, the population dynamics can be rewritten as $f_t(t,x) = \lambda f(t,x) [(\gamma x + \alpha x^2) - (\gamma M_1^f + \alpha M_2^f)]$, a function-valued ODE reminiscent of the replicator dynamics for continuous-strategy games studied by Bomze, Oechssler and Riedel, and others \citep{bomze1990dynamical,oechssler2001evolutionary,oechssler2002dynamic}. In this framework, the typical long-time behavior is concentration of $f(t,x)$ upon a delta-distribution at the maximizer of average group payoff $G(x)$ in $[0,1]$. 

\par The combination of the advection and nonlocal terms in Equation \ref{eq:pdeparam} characterizes the conflict between the tendency of within-group competition to favor defection and the tendency of between-group competition to favor groups with a majority of cooperators. In particular, we are interested in exploring the balance between these effects for varying levels of the relative rate of within-group and between-group birth and death events, $\Lambda$, as well as the relative strengths of selection at the two levels, $w_I$ and $w_G$. 

\subsection{Properties of PDE}

Because we derive the PDE of the multilevel system from a ball-and-urn process, it important to check that the limiting differential equation preserves normalization of the probability density over population states. Assuming that $\int_0^1 f(t,x) \dx = 1$ at a given time $t$, integrating both sides of Equation \ref{eq:pdeparam} with respect to $x$ from $0$ to $1$ yields \begin{align*} \dsdel{}{t} \left[\ds\int_0^1 f(t,x) \dx \right] %
&= \int_0^1 \left\{ \dsdel{}{x} \left[ x(1-x) (\beta + \alpha x ) f(t,x) \right] + \lambda f(t,x) \left[ \gamma x + \alpha x^2 - \left( \gamma M_1^f + \alpha M_2^f \right) \right]\right\} \dx  \\ &= %
x(1-x) (\beta + \alpha x ) f(t,x)\bigg|_0^1 + \lambda \left[ \left( \gamma M_1^f - M_2^f\right) - \left( \gamma M_1^f - M_2^f\right)  \int_0^1 f(t,x) \dx \right]  = 0 \end{align*}
and we can deduce that $\int_0^1 f(t,x) \dx = \int_0^1 f(0,x) \dx$ for all $t \geq 0$. In other words, if $f(0,x)$ is normalized, then solutions $f(t,x)$ of equation \ref{eq:pdeparam} remain normalized at all later times $t$. 

To study the behavior of Equation \ref{eq:pdeparam}, it helps to introduce the associated PDE \begin{equation} \label{eq:pdehgeneral} \dsdel{f(t,x)}{t} = \dsdel{}{x} \left[ x(1-x) (\beta + \alpha x ) f(t,x) \right] + \lambda f(t,x) \left[ \gamma x + \alpha x^2 -  h(t) \right]  \end{equation} where $h(t)$ is a general function of time. In analogy with Lemma 6 of Luo and Mattingly \citep{luo2017scaling}, we see that the only function $h(t)$ for which solutions of Equation \ref{eq:pdehgeneral} remain normalized is $h(t) = \gamma M_1^f(t) + \alpha M_2^f(t)$. Using the method of characteristics, we obtain the %
representation formula for solutions of Equation \ref{eq:pdehgeneral}
\begin{equation} \label{eq:solutionalongcharacteristics} f(t,x) = f_0(x_0(t,x)) \exp \left[\beta t + \int_0^1 \left\{ \left(\lambda \gamma - 2 (\alpha + \beta) \right) x(s)  + \left( \lambda + 3\right) \alpha x(s)^2 \right\} ds - \int_0^1 h(s) \dsint \right] \end{equation}
where $x_0 := x_0(t,x)$ can be found by solving the replicator dynamics backwards in time. 
From this representation formula, we see that non-negativity of the initial distribution $f(0,x) = f_0(x)$  on $[0,1]$ implies non-negativity of $f(t,x)$ on $[0,1]$ at subsequent times $t$ for which our solution in Equation \ref{eq:solutionalongcharacteristics} exists. 

Next we can use a contraction-mapping argument analogous to that of Dawidowicz and Loskot for nonlocal transport equations on spatial domain $[0,\infty)$ \citep{dawidowicz1986existence} to demonstrate that local existence of solutions to Equation \ref{eq:pdehgeneral} implies local existence to solutions of Equation \ref{eq:pdeparam}. By denoting $f^0(t,x)$ the solution of  \ref{eq:solutionalongcharacteristics}, we can construct a subsequent $f^1(t,x)$ as a solution via the method of characteristics to Equation \ref{eq:pdehgeneral} with $h(t)$ chosen as $h(t) = \gamma \int_0^1 y f^0(t,y) \dy + \alpha \int_0^1 y^2 f^0(t,y) \dy$. Extending this scheme, we can make use of the Banach fixed point theorem to guarantee that $f^n(t,x) \to f(t,x)$ as $n \to \infty$, where $f(t,x)$ is the unique non-negative solution to Equation \ref{eq:pdeparam}.  %
Then, using our representation formula with $h(s) = \gamma M_1^f(s) + \alpha M_2^f(s)$ and observing that $\beta - 2 (\alpha + \beta) x(s) + 3 \alpha x(s)^2 \leq |\beta| + 2 |\alpha + \beta| + 3 |\alpha|$, $\gamma (x(s) - M_1^f(s)) \leq |\gamma|$ and $ \alpha (x(s)^2 - M_2^f(s)) \leq |\alpha|$, we can estimate that solutions of Equation \ref{eq:pdeparam} with bounded and density-valued initial data satisfy  \begin{equation} \label{eq:ftxbound} f(t,x) \leq  \left(\sup_{x_0 \in [0,1]} f_0(x_0) \right) \exp\left[ \left\{(5 + \lambda) |\alpha| +  3 |\beta| +\lambda |\gamma| \right\}t \right]  < \infty \end{equation} for $x \in [0,1]$ and for all finite $t$, so our solutions obtained via the method of characteristics exist and remain density-valued for finite times $t$. 

\subsubsection{Weak, Measure-Valued Formulation}

In the absence of between-group competition, the expected behavior of Equation \ref{eq:pdeparam} is for all of the probability to accumulate as a delta-function near the stable steady-state of the within-group replicator dynamics. We formalize this intuition by considering a weak, measure-valued formulation of our PDE, and use this to prove some of our main results on the time-dependent behavior of the multilevel system.

Multiplying both sides of Equation \ref{eq:pdeparam} by test function $\psi(x) \in C^1[0,1]$ and integrating, we get
$$ \dsddt \ds\int_0^1 \psi(x) f(t,x) \dx = \ds\int_0^1 \psi(x)  \left\{ \dsdel{}{x} \left[x(1-x) (\beta + \alpha x) f(t,x) \right]  + \lambda f(t,x) \left[\gamma x + \alpha x^2 - \left(\gamma M_1^f + \alpha M_2^f\right)  \right]  \right\} \dx $$
Integrating the first term by parts, we have that 
\begin{align*} \dsddt \ds\int_0^1 \psi(x) f(t,x) \dx &= - \ds\int_0^1 \left\{\dsdel{\psi(x)}{x}  \left[x(1-x) (\beta + \alpha x)\right] +  \lambda   \psi(x) \left[\gamma x + \alpha x^2 - \left(\gamma M_1^f + \alpha M_2^f\right)  \right] \right\} f(t,x) \dx \end{align*}
 For the cumulative distribution function $F(t,x)$ corresponding to the density $f(t,x)$, we note that $f(t,x) \dx = F_x(t,x) \dx = d F(t,x)$. In the sense of a Stietjes integral, we can then write
\begin{align*} \dsddt \ds\int_0^1 \psi(x) \mathrm{d}F(t,x) &= - \ds\int_0^1 \left\{\dsdel{\psi(x)}{x}  \left[x(1-x) (\beta + \alpha x)\right] +  \lambda   \psi(x) \left[\gamma x + \alpha x^2 - \left(\gamma M_1^F + \alpha M_2^F\right)  \right] \right\} \mathrm{d}F(t,x) \end{align*} 
Denoting the probability measure associated with $f(t,x)$ and $F(t,x)$ by $\mu_t(dx)$, we have the equation
 \begin{equation} \label{eq:measurepde} \dsddt \ds\int_0^1 \psi(x) \mu_t(dx) =  \ds\int_0^1 \left\{- \dsdel{\psi(x)}{x}  \left[x(1-x) (\beta + \alpha x)\right] + \lambda  \psi(x) \left[\gamma x + \alpha x^2 - \left(\gamma M_1^{\mu} + \alpha M_2^{\mu}\right)  \right] \right\} \mu_t(\dx) \end{equation}
 where we have denoted the $j$th moments of the probability distribution by $M_j^f$, $M_j^F$, and $M_j^{\mu}$ to match the description of the probability distribution in terms of its descriptions in terms of a density $f(t,x)$, cumulative distribution function $F(t,x)$ or measure $\mu_t(dx)$. 

 We say that the flow of measures $\{\mu_t(\dx)\}_{t \in [0,T]}$ is a weak solution of \ref{eq:measurepde} with given initial data $\mu(\dx)$ if for every test function $\psi(x) \in C^1([0,1])$, $\langle \psi, \mu_t(\dx) \rangle$ is differentiable in time, satisfies Equation \ref{eq:measurepde} for all $t \in [0,T]$, and fulfills the initial condition $\mu_0(dx) = \mu(\dx)$.
 \begin{remark} With this measure-valued formulation, we can see that the measures $\delta(x)$, $\delta(x-1)$, and $\delta(x - \tfrac{\beta}{\alpha})$, the delta-functions at the steady-states for within-group dynamics, are also steady-states for Equation \ref{eq:measurepde}. For the first term, we note that $x(1-x)(\beta + \alpha x)$ vanishes at these points, and for the second term, with any measure $\mu_t(\dx) = \delta(x-a)$ (with corresponding $M_1^f = a $ and $M_2^f = a^2$) satisfies $$  \ds\int_0^1 \psi(x) \left[\gamma x + \alpha x^2 - \left(\gamma M_1^f + \alpha M_2^f\right)  \right] \mu_t(\dx)  = \psi(a) \left[ \left(\gamma a +\alpha a^2\right) - \left(\gamma M_1^f + \alpha M_2^f \right) \right] = 0 $$   \end{remark}
 
We can break down the dynamics of our multilevel system into changes caused by within-group competition and those caused by between-group events. Individual-level  competition within a single group is governed by the replicator dynamics given by Equation \ref{eq:replicator}, whose solution given time $t$ and initial-condition $x_0$ we call $\phi_t(x_0) := x(t,x_0)$. We can describe the effect of within-group competition in the group-structured population using the push-forward measure of the initial distribution $\mu_0(\dx)$ under the dynamics $\phi_t(x)$ \citep{luo2017scaling,evers2016mild} using the following equivalent notation  \[ P_t \mu_0(\dx) = \phi_t \# \mu_0(\dx) = \left[ \mu_0 \circ \phi_t^{-1} \right](\dx)  \]
We can illustrate the effects of within-group dynamics by testing the push-forward measure against the function $\psi(x)$, which yields
\[ \int_0^1 \psi(x) P_t \mu_0(\dx) = \int_0^1 \psi(x) \left[\mu_0 \circ \phi_t^{-1} \right] (\dx) = \int_0^1 \psi(\phi_t(x)) \mu_0(\dx) \]
In particular, this representation tells us how the probability distribution evolves as $\mu_t(\dx) = P_t \mu_0(\dx)$ in the case where $\lambda = 0$ in Equation \ref{eq:measurepde}, when the multilevel system is governed by within-group competition alone.

\myindent We can also describe the effect of between-group competition using the formula below
\begin{equation} \label{eq:wtxgeneral} w_t(x) = \exp \left( \lambda \int_0^t \left[ \left( \gamma x(s,x_0) + \alpha x(s,x_0)^2\right) - \left( \gamma M_1^{\mu} + \alpha M_2^{\mu} \right) \right] \dsint   \right) \end{equation}
Combining the effects of within-group and between-group competition, we can represent the evolution of the probability distribution using the implicit formula 
\begin{equation} \label{eq:mutdx} \mu_t(\dx) = w_t(x) P_t \mu_0(\dx) =  \exp \left( \lambda \int_0^t \left[ \left( \gamma x(s,x_0) + \alpha x(s,x_0)^2\right) - \left( \gamma M_1^{\mu} + \alpha M_2^{\mu} \right) \right] \dsint   \right)  P_t \mu_0(\\dx) \end{equation} 
where we note that this is an implicit description because the moments $M_1^{\mu}$ and $M_2^{\mu}$ depend on the distribution $\mu_t(\dx)$. In future work, we will address the existence and uniqueness of measure-valued solutions $\mu_t(\dx)$ to Equation \ref{eq:measurepde} given initial distribution $\mu_0(\dx) = \mu(\dx)$, and we justify the use of this representation formula to describe the evolution of $\mu_t(\dx)$. Given this formula, we can describe the behavior of our multilevel system by testing $\psi(x)$ against $\mu_t(\dx)$ as follows
\begin{equation} \label{eq:mutpushforward} \int_0^1 \psi(x) \mu_t(dx) = \int_0^1 \psi(x) w_t(x) P_t \mu_0(\dx) = \int_0^1 \psi(x) w_t(x) \left[ \mu_0 \circ \phi_t^{-1}\right] (\dx) = \int_0^1 \psi(\phi_t(x)) w_t(\phi_t(x)) \mu_0(\dx) \end{equation}
\subsubsection{Possible Group-Level Reproduction Functions}

An important factor for the long-time behavior of Equation \ref{eq:pdeparam} is where the group-level reproduction rate $G(x) = \gamma x + \alpha x^2$ is maximized. If $x^* = \argmax_{x \in [0,1]}\left(G(x)\right) = 1$, then all-cooperator groups are most favored at the between-group level, and would be selected to dominate the population in the absence of within-group selection. If $x^* \in (0,1)$, then between-group selection instead most favors groups with proportions of $x^*$ cooperators and $1-x^*$ defectors, and acts to promote coexistence between cooperators and defectors in the long-run. 

For the Hawk-Dove game, %
$\gamma  = (S-P) + (T-P) > 0$ and %
$\alpha = (T-R) + (S-P) < 0$. Because $G(x)$ has a local maximum  at $x^* = - \frac{\gamma}{2 \alpha}$ when $G'(x) = \gamma + 2 \alpha = 0$ , we see from $\gamma > 0$ and $\alpha < 0$ that 
 \begin{displaymath}
   x^* = \left\{
     \begin{array}{cr}
       - \ds\frac{\gamma}{2 \alpha} & :  \gamma < - 2 \alpha \\
       1 & : \gamma \geq- 2 \alpha
     \end{array}
   \right. \Longrightarrow
   x^* = \left\{
     \begin{array}{cr}
        \ds\frac{S + T - 2P}{2 (-R + S + T - P)} & :  2R < T + S \\
       1 & :  2R \geq T  + S
     \end{array}
   \right. 
\end{displaymath} 
For the Prisoner's Dilemma, we can't make a similar definitive statement about the signs of $\gamma$ and $\alpha$. Instead, we consider two special cases in which $\gamma > 0$ and either $\alpha = -1$ or $\alpha = 0$. In these two special cases, we see that the group fitness function $G(x)$ has the following behavior
\begin{itemize}
\item[Case I:] For $\gamma > 0$ and $\alpha = -1 $,  then $G(x) = \gamma x - x^2$ and we have a similar situation to the Hawk-Dove game, where $x^* = \left\{
     \begin{array}{rr}
       - \ds\frac{\gamma}{2} & :  \gamma < 2 \\
       1 & : \gamma \geq  2 
     \end{array}
   \right.$.
\item[Case II:] If $\gamma > 0$ and $\alpha = 0$, we have $G(x) = \gamma x$ and recover a scaled version of the Luo-Mattingly model governed by Equation \ref{eq:luomattingly}. For this case, we have $x^* = 1$, $\forall \gamma > 0$. 
\end{itemize}
Having explored the value of $x^*$ for the Hawk-Dove games and the Prisoner's Dilemmas we consider in this paper, we will see in subsequent sections that whether $x^* = 1$ or $x^* \in (0,1)$ will determine a change in qualitative behavior in the limit when $\lambda = \frac{w_G}{w_I} \Lambda \to \infty$. Notably, this is the limit in which $w_I \ll w_G$, selection strength in between-group competition is weak relative to selection strength for within-group competition.
\section{Prisoner's Dilemma} \label{sec:PD}

In this section, we consider the multilevel dynamics when strategic interactions consist of Prisoner's Dilemma games. %
 In Section \ref{sec:PDsolvable}, we consider a special family of payoff matrices for the Prisoner's Dilemma for which the replicator dynamics are analytically solvable. With solvable replicator dynamics, we use the method of characteristics to study the long-time behavior of our multilevel system, as used for the frequency-independent case by Luo and Mattingly \citep{luo2017scaling}. In Section \ref{sec:PDsteadystates}, we characterize the steady-states for our multilevel PDE and explore the behavior of steady-states in the limit of large $\lambda$. 

\subsection{Prisoner's Dilemma with Solvable Replicator Dynamics} \label{sec:PDsolvable}

If we consider a Prisoner's Dilemma with $T = R + 2$, $P = S + 1$, then we have that $\alpha = -1$ and $\beta = 1$, and our payoff matrix corresponds to a Case I PD. This family of payoff matrices takes the form
\begin{equation} \label{eq:specialpdpayoffmatrix}
\begin{blockarray}{ccc}
& C & D \\
\begin{block}{c(cc)}
C & \gamma + P - 1 &    P - 1 \\
D & \gamma + P + 1 & P \\
\end{block}
\end{blockarray}
\end{equation}
where $\gamma > 1$ guarantees that the corresponding game is actually a PD (so that $R > P$ or $\gamma + P - 1 > P$). We note that, in this case, $G(x) = \gamma x - x^2$, which is maximized at $x^* = \min(\tfrac{\gamma}{2},1) \in (\tfrac{1}{2},1]$ for $\gamma > 1$. In other words, group payoff is maximized by groups with a majority of cooperators for $\gamma > 1$ and is maximized by all-cooperator groups for $\gamma \geq 2$. Cooperators and defectors become perfect complements for average group payoff $G(x) = x(1-x)$, when $\gamma = 1$ and the game is no longer a Prisoner's Dilemma.

The replicator dynamics for this game is $$ \label{eq:specialPDrep} \dsddt{x(t)} = - x (1-x) (1+x)  $$ which can be solved forwards in time for $x(t)$ given initial condition $x_0$ and $t$ as \begin{equation} \label{eq:specialPDchar} x(t,x_0) = \phi_t(x_0)  = \frac{x_0}{\sqrt{(1-x_0)^2 e^{2t} + x_0^2}} \end{equation} %
We note that $\lim_{t \to \infty} x(t) = 0$, so the characteristics converge to the Nash equilibrium for the Prisoner's Dilemma. We can also solve the replicator dynamics backwards in time for $x_0(x,t)$ given current $t$ and $x(t)$ with
\begin{equation} \label{eq:specialPDbackwardschar} x_0(t,x(t)) = \phi_t^{-1}(x) = \frac{x}{\sqrt{(1-x^2)e^{-2t} + x^2}}  \end{equation}
Because the replicator dynamics are exactly solvable for payoff matrices of this form, we can use the method of characteristics in the manner employed by Luo \citep{luo2014unifying} and Luo and Mattingly \citep{luo2017scaling} to explore the long-time behavior of solutions to Equation \ref{eq:pdeparam}. For such payoff matrices, we have the simpler PDE 
\begin{equation} \label{eq:specialPDPDE} \dsdel{f(t,x)}{t} = \dsdel{}{x} \left[x(1-x)(1+x) f(t,x)  \right] + \lambda f(t,x) \left[\left(\gamma x - x^2 \right) - \left( \gamma M_1^f - M_2^f \right) \right] \end{equation}
The contributions to the solution due solely to between-group competition are given by $$w_t(x) := \exp \left( \ds\lambda \left(\gamma \int_0^t x(s,x_0) ds  - \ds\int_0^t x(s,x_0)^2 ds - \ds\int_0^t h(s) ds  \right)\right) $$ As shown in the appendix, we can use the forward and backward solutions of the replicator dynamics to compute that  %
\begin{align} \label{eq:xint} \ds\int_0^t x(s,x_0) ds &= t + \frac{1}{2} \log\left(\frac{1-x}{1+x} \right) + \frac{1}{2} \log \left( e^{-2t} + \frac{2x^2 + 2 x \sqrt{e^{-2t} \left(1 - x^2 \right) + x^2 }}{(1-x)(1+x)} \right) \\ \label{eq:x2int} \ds\int_0^t x(s,x_0)^2 ds &= t - \frac{1}{2} \log\left(e^{-2t}(1-x^2) + x^2 \right) \end{align} and using the values of these integrals we see that $$w_t(x) = \left(\frac{1-x}{1+x} \right)^{\frac{\lambda \gamma}{2}} \left( e^{-2t} + \frac{2x^2 + 2 x \sqrt{e^{-2t} \left(1 - x^2 \right) + x^2 }}{(1-x)(1+x)} \right)^{\frac{\lambda \gamma}{2}} \left( e^{-2t} (1 - x^2) + x^2 \right)^{-\frac{\lambda}{2}} e^{\lambda ( \gamma -1) t - \int_0^1 h(s)ds}  $$

We can isolate the component of $w_t(x)$ that depends only on time by writing $w_t(x) = g_t(x) e^{\lambda (\gamma - 1) t - \int_0^1 h(s) ds}$, where %
 \begin{equation} \label{eq:gtx} g_t(x) = \left(\frac{1-x}{1+x} \right)^{\frac{\lambda \gamma}{2}} \left( e^{-2t} + \frac{2x^2 + 2 x \sqrt{e^{-2t} \left(1 - x^2 \right) + x^2 }}{(1-x)(1+x)} \right)^{\frac{\lambda \gamma}{2}} \left( e^{-2t} (1 - x^2) + x^2 \right)^{-\frac{\lambda}{2}} \end{equation}%
To describe the long-time behavior of our multilevel system, it will help to use the approach of Luo and Mattingly \citep{luo2017scaling} and characterize the behavior of the tail of our initial distribution using the Hölder exponent near $x = 1$, which is defined as follows.

\begin{definition} \label{def:Holderexponent}  The Hölder exponent of our measure $\mu_t[0,x]$, near the endpoint $x = 1$ is given by $$\theta_t = \ds\inf \left\{ \Theta \geq 0 \:\bigg| \: \exists C > 0 \: \: \mathrm{s.t.} \: \: \ds\lim_{x \to 0} \ds\frac{\mu_0[1-x,1]}{x^{\Theta}}  = C\right\} $$ \end{definition}

\begin{example} \label{ex:holderfamily}  A family of densities with Hölder exponent $\theta$ is given by $f(x) = \theta (1-x)^{\theta - 1} $. For $\theta = 1$, we recover the uniform density $f(x) = 1$. For $\theta < 1$, the density of groups blows up near full cooperation ($x=1$), while for $\theta > 1$, the density tends to $0$ near full cooperation.  \end{example}

Intuitively, we can think of the Hölder exponent near $x=1$ as describing how strongly or weakly the probability distribution is concentrated near all-cooperator groups. We see that larger $\theta$ corresponds to weaker concentration of probablity near all-cooperator groups, so stronger between-group selection is required in order to sustain groups with many cooperators. We will see in the propositions below that the ability for multilevel selection to maintain cooperation in the population is tied to competition between relative strengths $\lambda$ of between-group and within-group selection and the concentration of near-all-cooperator groups $\theta$. This phenomenon is seen in the frequency-independent model of Luo and Mattingly \citep{luo2017scaling}, and has been shown as important for emergence of cooperation in simulations of the game theoretic models of protocell evolution by Markvoort et al \citep{markvoort2014computer}.

Now we characterize the long-time behavior of our multilevel system. To develop intuition for the role of the \holder exponent of the initial distribution, we can solve our multilevel system for our family of initial densities of the form $\theta(1-x)^{\theta - 1}$ for $\theta > 0$ from Example \ref{ex:holderfamily}. 

\begin{example} \label{ex:exactsolution} For the family of payoff matrices under consideration, for which $\alpha = -1$ and $\beta = 1$, we find from Equation \ref{eq:solutionalongcharacteristics} that solutions of the multilevel dynamics satisfy
\begin{equation} \label{eq:specialPDcharacteristics} f(t,x) = f_0(x_0(t,x)) \exp \left[t + \int_0^t \left\{ \lambda \gamma x(s)  - \left( \lambda + 3\right) \alpha x(s)^2 \right\} ds - \int_0^1 h(s) ds \right] \end{equation}
We can use Equation \ref{eq:specialPDbackwardschar} to see that the initial distribution is pushed forward by within-group dynamics as follows
\begin{align*} 
f_0(x_0(t,x)) &= f_0 \left( \frac{1}{\sqrt{\left( \frac{1}{x^2} - 1\right) e^{-2t} + 1}}  \right) = \theta \left( 1 - \frac{1}{\sqrt{\left( \frac{1}{x^2} - 1\right) e^{-2t}+ 1}}  \right)^{\theta - 1} \\ &= \theta \left( \frac{e^{-2t} (1-x^2) + x^2}{e^{-2t} (1-x^2) + x^2} - \frac{x \sqrt{e^{-2t} (1-x^2) + x^2 }}{e^{-2t} (1-x^2) + x^2} \right)^{\theta - 1} \\ &= \theta e^{2(1 - \theta )t} \left(  \left(1 - x^2 \right) + e^{2t} \left(x^2 - x \sqrt{e^{-2t}\left(1 - x^2\right) + x^2} \right) \right)^{\theta - 1} \left(e^{-2t} (1-x^2) + x^2 \right)^{1 - \theta} \end{align*}
Combining this with the expressions for $\int_0^t x(s) ds$ and $\int_0^t x(s)^2 ds$ given by Equations \ref{eq:xint} and \ref{eq:x2int}, using the expression $g_t(x)$ from Equation \ref{eq:gtx}, and recalling that solutions to Equation \ref{eq:specialPDPDE} are normalized, we see that our solution is given by 
\begin{equation} \label{eq:PDspecialsolution} f(t,x) = \theta Z_f^{-1}e^{\left[\lambda  \left(\gamma - 1 \right) - 2 \theta\right] t} g_t(x) \left(  \left(1 - x^2 \right) + e^{2t} \left(x^2 - x \sqrt{e^{-2t}\left(1 - x^2\right) + x^2} \right) \right)^{\theta - 1} \left(e^{-2t} (1-x^2) + x^2 \right)^{1 - \theta}  \end{equation}
where $Z_f$ is a normalizing constant (so that $\int_0^1 f(t,y) \dy = 1$).
Heuristically, we see that if $\lambda (\gamma - 1) > 2 \theta$, then $f(t,x) \to Z_f^{-1} x^{\lambda (\gamma - 1) - 2 \theta - 1} (1-x)^{\theta - 1} (1 + x)^{- \lambda \gamma + \theta - 1}$. However, if $\lambda (\gamma - 1) < 2 \theta$, then the $\exp\left(\left[ \lambda (\gamma - 1) - 2 \theta \right] t\right)$ term will cause the decay of probability density for $x \in (0,1)$, and, as we will soon see, $f(t,x) \rightharpoonup \delta(x)$ (groups concentrate at the all-defector within-group equilibrium of the Prisoner's Dilemma). We also note that one can use Definition \ref{def:Holderexponent} to verify that the \holder exponent is preserved in time for solutions of Equation \ref{eq:specialPDPDE} for the given family of initial data, as shown in Appendix \ref{sec:Holderpreserved}.
 \end{example}

Now we are ready to characterize the long-time behavior of the multilevel PD system. The qualitative behavior is divided into two regimes by a critical level of relative strength of between-group selection ($\lambda^* := \frac{2 \theta}{\gamma - 1}$) which depends on the payoff matrix through $\gamma$ and on the initial data through the \holder exponent $\theta$ near $x=1$. For $\lambda < \lambda^*$, the distribution of groups converge to a delta-function at full-defection ($\delta(x)$) as $t \to \infty$. When $\lambda > \lambda^*$, the distribution of groups converges to a density supported at all group types as $t \to \infty$, so groups with any mix of cooperators and defectors are sustained at steady-state. We note that the payoff of all-cooperator groups is $G(1) = \gamma - 1$ and that the condition required for the survival of cooperation is $\lambda (\gamma - 1) = \lambda G(1) > 2 \theta$. This means that the ability to sustain cooperation in our multilevel selection model requires a sufficiently capable combination of between-group selection strength $\lambda$ and payoff for all-cooperator groups $G(1)$ relative to the extent to which the initial distribution fails to concentrate near all-cooperator groups $\theta$. This is reflective of the inevitable decrease of cooperators within groups due to individual birth and death events, and that the main opportunity to delay the within-group march towards defection is the reproduction of all-cooperator groups.

\begin{proposition} \label{prop:deltaspecialpd} Suppose our initial population distribution has Hölder exponent $\theta$ near $x = 1$. If $\lambda(\gamma - 1) < 2 \theta$, then $\mu_t(dx) \rightharpoonup \delta(x)$ as $t \to \infty$. \end{proposition}

Here, ``$\rightharpoonup$'' denotes weak convergence of the probability measures $\{\mu_t(dx)\}_{t \geq 0}$ as $t \to \infty$ to a limit $\mu_{\infty}(dx)$, or that for any test function $\psi(x) \in C^1[0,1]$, $\int_0^1 \psi(x) \mu_t(dx) \to \int_0^1 \psi(x) \mu_{\infty}(dx)$ as $t \to \infty$. 

\begin{proof}

We will show that, for any continuous function $\psi(x)$, $\int_0^1 \psi(x) \mu_t(dx) \to \psi(0)$ as $t \to \infty$. Because $\psi(\cdot)$ is continuous, we know that $\forall \epsilon > 0$, $\exists \delta > 0$ such that $|\phi(x) - \phi(0)| < \epsilon$ when $x \in [0,\delta]$. Because $\mu_t(dx)$ is a probability distribution, we can say that \begin{align*}\bigg| \ds\int_0^1 \psi(x) \mu_t(dx) - \psi(0) \bigg| %
&\leq \int_0^{\delta} |\psi(x) - \psi(0) | \mu_t(dx) + \ds\int_{\delta}^1 |\psi(x) - \psi(0)| \mu_t(dx) \\ &< \epsilon + 2 ||\psi||_{\infty} \ds\int_{\phi_t^{-1}(\delta)}^1 w_t(\phi_t(y)) \mu_0(dy)   \end{align*}
We rewrite $$g_t(x) = \underbrace{\left( \frac{1-x}{1+x} e^{-2t} + \frac{2x^2 + 2 x \sqrt{e^{-2t} \left(1 - x^2 \right) + x^2 }}{(1+x)^2} \right)^{\frac{ \lambda \gamma}{2}} }_{\leq \Huge{7}^{(\lambda \gamma) / 2}}
  \underbrace{\left( e^{-2t} (1 - x^2) + x^2 \right)^{-\frac{\lambda}{2}} }_{\leq \large{\delta^{\lambda}} \: \: \mathrm{for} \: \: x \in [\delta,1] }  $$ and so  $\exists M > 0$ such that $\forall x \in [\phi_t^{-1}(y),1]$, $g_t(\phi_t(y)) \leq M$ and $w_t(\phi_t(y)) \leq M e^{\lambda (\gamma - 1) t - \int_0^t h(s) ds}$.

Further, for $h(s) = \gamma M_1^{\mu}(s) - M_2^{\mu}(s)$, we note that $\gamma > 1$ for the Prisoner's Dilemma and that $M_2^{\mu} \leq M_1^{\mu}$ because $x^2 \leq x$ on $\mathrm{supp}(\mu_t(dx)) \subset [0,1]$.
Thus $h(s) \geq M_1^f(s) - M_2^f(s) > 0$, so $e^{- \lambda \int_0^t h(s) ds} \leq 1$ for $t \geq 0$, and therefore $w_t(\phi_t(y)) \leq M e^{\lambda (\gamma - 1) t}$ for each $t \geq 0$ and for each $x \in [\phi_t^{-1}(\delta),1]$.  Therefore
$$\bigg| \ds\int_0^1 \psi(x) \mu_t(dx) - \psi(0) \bigg| \leq \epsilon + 2 M ||\psi||_{\infty} e^{\lambda (\gamma - 1) t} \ds\int_{\phi_t^{-1}(\delta)}^{1} \mu_0(dx) = \epsilon + 2 M ||\psi||_{\infty} e^{\lambda (\gamma - 1) t} \mu_0[\phi_t^{-1}(\delta),1]   $$
Using $\phi_t^{-1}(x) = \left(\left( \tfrac{1}{x^2} - 1 \right) e^{-2t} + 1\right)^{-1/2} \in [0,1]$, we see that $$\phi_t^{-1}(x) \geq \frac{1}{\left( \frac{1}{x^2} - 1 \right) e^{-2t} + 1} = 1 - \frac{\left( \frac{1}{x^2} - 1 \right) e^{-2t}}{ \left( \frac{1}{x^2} - 1 \right) e^{-2t} + 1} \geq 1 - \left(\frac{1}{x^2} - 1 \right) e^{-2t} \: \: \textnormal{ because } \: x \in [0,1]$$
Now, for $x \in [\delta,1]$ and $\forall t \geq 0$, we see that $\exists D > 0$ such that $\phi_t^{-1}(x) \geq 1 - D d^{-2t}$, %
so $\mu_0[\phi_t^{-1}(x),1] \leq \mu_0[1 - D e^{-2t},1]$. Using the assumption that $\lim_{x \to 0} x^{- \theta} \mu_0[1-x,1] = C $ for some constant $C \in ]0, \infty[$, we can deduce that $\mu_0[1 - D e^{-2t},1]  \leq C D\left(e^{-2t}\right)^{\theta}$. Then, because $\lambda (\gamma - 1) < 2 \theta$, we have that
 \begin{align*} \bigg| \ds\int_0^1 \psi(x) \mu_t(dx) - \psi(0) \bigg| &\leq  
\epsilon +  2 MCD  ||\psi||_{\infty} e^{[\lambda (\gamma - 1) - 2 \theta] t}%
<  \epsilon \: \: \mathrm{as} \: \: t \to \infty \end{align*} and therefore $\int_0^1 \psi(x) \mu_t(dx) \to \psi(0)$ at $t \to \infty$. 
\end{proof}
\begin{proposition} \label{prop:steadystatespecialpd} If $\lambda (\gamma - 1) > 2 \theta$, then $\mu_t(dx) \rightharpoonup \mu_{\infty}^{\theta}(dx) = Z_f^{-1} x^{\lambda (\gamma - 1) - 2 \theta - 1} (1-x)^{\theta -1} (1+x)^{-\frac{\lambda \gamma}{2} + \theta - 1} \dx$, where $Z_f$ is a normalizing constant such that $\int_0^1 \mu_{\infty}^{\theta}(dx) = 1$. 
 \end{proposition}

Sample steady-states for given initial distributions and various values of $\lambda$ and $\theta$ are given in Figure \ref{fig:noghostdensity} for $\gamma = 2.5$ and in Figure \ref{fig:ghostdensity} for $\gamma = 1.5$. In Figure \ref{fig:noghostdensity}, we see that steady-state densities feature more cooperators as $\lambda$ increases, and that less relative selection strength $\lambda$ is required to sustain cooperation for initial conditions with $\theta = 1$ (left) than for initial conditions with $\theta = 2$ (right). In Figure \ref{fig:ghostdensity}, we see that groups with optimal average payoff cannot be sustained in large quantity at steady-state even for large relative between-group selection strength.

\begin{figure}[H]
    \centering
    \includegraphics[width=0.495\textwidth]{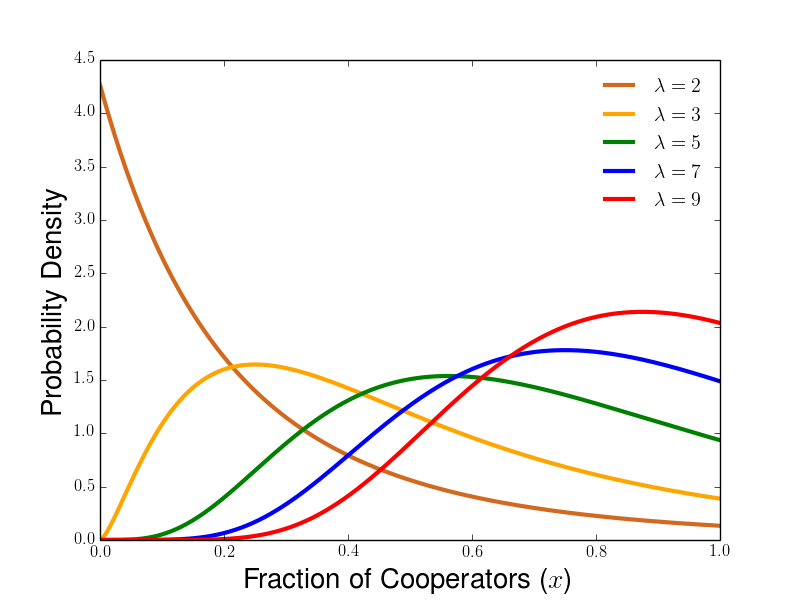}
    \includegraphics[width=0.495\textwidth]{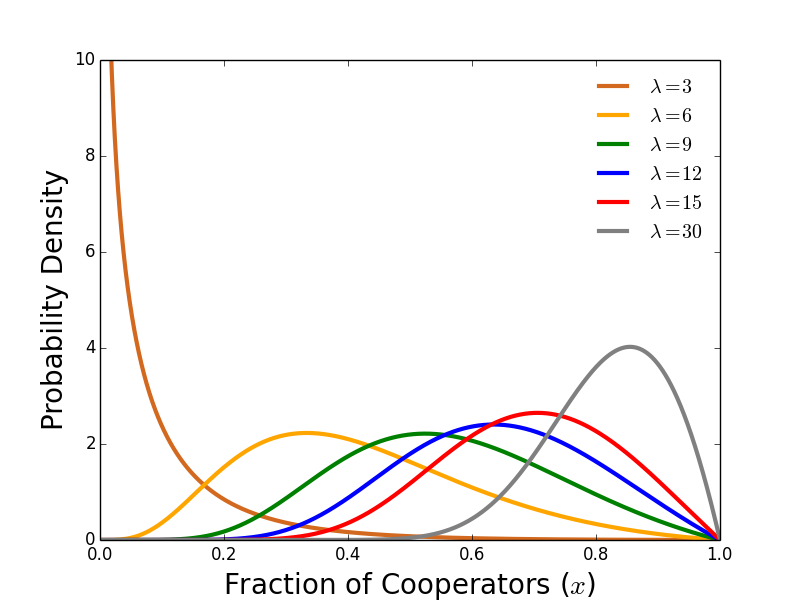}
    \caption{Steady state densities for various relative selection strengths $\lambda$ when $\gamma = 2.5$ ($G(x)$ maximized by $x=1$), computed from the result of Proposition \ref{prop:steadystatespecialpd}. (Left) Steady state corresponding to initial distribution with Hölder exponent $\theta_0 = 1$ near $x=1$. (Right) Steady state corresponding to initial distribution with Hölder exponent $\theta_0 = 1$ near $x=1$.}
    \label{fig:noghostdensity}
\end{figure}

\begin{figure}[H]
    \centering
    \includegraphics[width=0.495\textwidth]{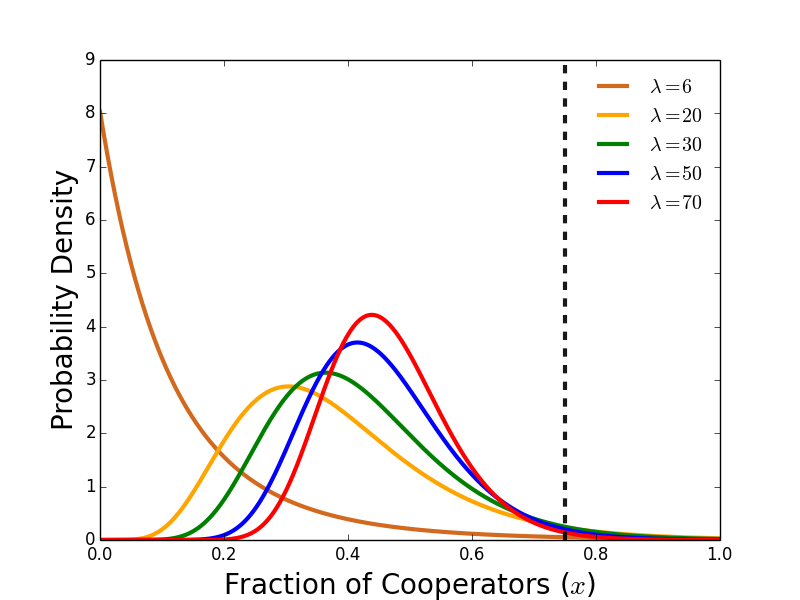}
    \includegraphics[width=0.495\textwidth]{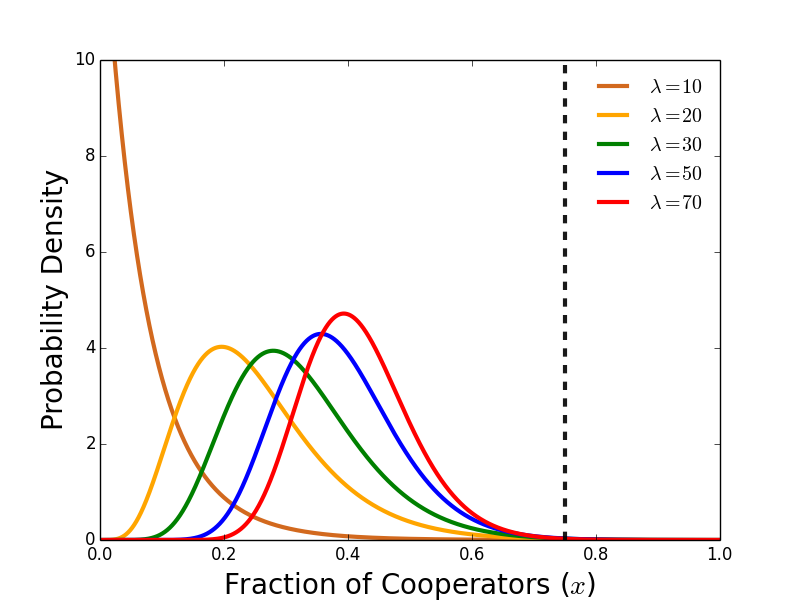}
    \caption{Steady state densities for various relative selection strengths $\lambda$ when $\gamma = 1.5$ ($G(x)$ maximized by $x=0.75$), computed from the result of Proposition \ref{prop:steadystatespecialpd}. (Left) Steady state corresponding to initial distribution with Hölder exponent $\theta_0 = 1$ near $x=1$. (Right) Steady state corresponding to initial distribution with Hölder exponent $\theta_0 = 1$ near $x=1$. Dotted line indicates that average group payoff $G(x)$ is maximized at $75$ percent cooperators.}
    \label{fig:ghostdensity}
\end{figure}

\begin{proof}  
For $\lambda (\gamma - 1) > 2 \theta$, we show, for any $\psi \in C([0,1])$, that there is a family of normalizing constants $\{Z_t\}_{t \geq 0}$ such that $Z_t \int_0^1 (\psi g_t) \circ (\phi_t(x)) d \mu_0(x) \to \int_0^1 \psi(x) f_{\theta}(x) \dx$. First we integrate by parts, obtaining $$\ds\int_0^1 [(\psi g_t) \circ \phi_t] (x) d \mu_0(x)  = \psi g_t F_0(1^-) - \psi g_t F_0(0^+) - \ds\int_0^1 \partial_x \left[(\psi g_t) \circ \phi_t  \right] (x) F_0(x) \dx  $$ where $F_0(x) :=  \mu_0[0,x]$ is the cumulative distribution function of $\mu_0$. Because $\phi_t(x)$ is continuous and satisfies $\phi_t(0) = 0$ and $\phi_t(1) = 1$, %
we know that $\int_0^1 \partial_x \left[(\psi g_t) \circ \phi_t  \right] (x) \dx = [\psi g_t](1^-) - [\psi g_t ](0^+) $. This allows us to rewrite our above equation as  $$\ds\int_0^1 [(\psi g_t) \circ \phi_t] (x) d \mu_0(x)  = [\psi g_t(1 -  F_0)](1^-) - [\psi g_t (1 - F_0](0^+) + \ds\int_0^1 \partial_x \left[(\psi g_t) \circ \phi_t  \right] (x) (1 - F_0(x)) \dx  $$ 

Because the group compositions are supported on $[0,1]$, $1 - F_0(1^-) = 0$, and we have from Equation \ref{eq:gtx} that $g_t(0^+) = e^{-\lambda (\gamma - 1)}$. This allows us to see that  $$[\psi g_t (1 - F_0)](0^+) = \psi(0) g_t(0) (1 - F_0(0^+)) = (1 - F_0(0^+)) \psi(0) e^{- \lambda (\gamma-1)t}] $$ Picking $Z_t = e^{2 \theta t}$, we see that $$\ds\lim_{t \to \infty} e^{2 \theta t} \left( \psi g_t (1 - F_0))(0^+)  \right) = \ds\lim_{t \to \infty} \left( (1 - F_0(0^+)) \psi(0) e^{- (\lambda (\gamma - 1) - 2 \theta)t} \right) \to 0 \: \: \mathrm{when} \: \: \lambda (\gamma - 1) > 2 \theta $$

For the integral term, we write $y = \phi_t(x)$ and see that $\partial_x \left[\psi (y) g_t(y) \right] = \partial_y[\psi(y) g_t(y)] \cdot \partial_x y$, so integrating with respect to $y$ gives \begin{align*}  \ds\int_0^1 \partial_x \left[(\psi g_t) \circ \phi_t  \right] (x) (1 - F_0(x)) \dx   %
&=   \ds\int_0^1 \partial_y \left[ \psi(y) g_t(y) \right]  ( 1 - F_0(\phi_t^{-1}(y)) \dy \end{align*}
Because $g_t(x) \to x^{\lambda (\gamma - 1)} (1 +x)^{- \lambda \gamma}$ as $t \to \infty$, we see that $$\ds\lim_{t \to \infty}\partial_x \left[\left(\psi g_t\right) \circ x \right]  = \partial_x \left[ x^{\lambda (\gamma -1)} (1+x)^{-\lambda \gamma} \psi(x) \right]  $$

We also observe that \begin{align*} 1 - F_0(\phi_t^{-1}(x)) &= \mu_0[\phi_t^{-1}(x),1] = \mu_0[1 - (1 - \phi_t^{-1}(x)),1] \\ &= \mu_0\left[1 - \frac{ \left( 1 - x^2 \right) e^{-2t} + x^2 -  x \sqrt{\left( 1 - x^2 \right) e^{-2t} + x^2}}{\left( 1 - x^2 \right) e^{-2t} + x^2},1 \right] \end{align*} Then, using our assumption on $\mu_0[1-x,1]$ for $x$ near $1$ to deduce that for large $t$, \begin{align*} 1 - F_0(\phi_t^{-1}(x)) &\approx %
C e^{-2\theta t} \left( \frac{ \left( 1 - x^2 \right)  + e^{2 t } \left( x^2 -  x \sqrt{\left( 1 - x^2 \right) e^{-2t} + x^2} \right)}{\left( 1 - x^2 \right) e^{-2t} + x^2} \right)^{\theta} \end{align*} We can also compute that %
$\ds\lim_{t \to \infty} e^{2 \theta t} \left[1 - F_0( \phi_t^{-1}(x))\right] \to C 2^{-\theta} \left((1-x)(1+x) \right)^{\theta} x^{-2 \theta}$, and
$$ e^{2 \theta t} \ds\int_0^1 \left[(\psi g_t) \circ \phi_t(x) \right] \mu_0(dx) \to \frac{C}{2^{\theta}} \ds\int_0^1 \partial_x \left[ x^{\lambda (\gamma -1)} (1+x)^{-\lambda \gamma} \psi(x) \right]  \left(\frac{(1-x)(1+x)}{x^2} \right)^{\theta} \dx $$ as $t \to \infty$. We can then integrate the righthand side by parts to see that $$e^{2 \theta t} \ds\int_0^1 \left[(\psi g_t) \circ \phi_t(x) \right] \mu_0(dx) \to \frac{2  C}{2^{\theta-1}} \ds\int_0^1  x^{\lambda (\gamma - 1) - 2 \theta - 1} (1-x)^{\theta - 1} (1+x)^{-\lambda \gamma + \theta - 1} \psi(x) \dx   \eqno{\qedhere} $$
 \end{proof}
 
 \begin{remark} In the case where $\lambda (\gamma - 1) > 2 \theta$, we can use Definition \ref{def:Holderexponent} to see that the steady-state density $f_{\theta}(x)$ also has \holder exponent of $\theta$ near $x = 1$.
Because our steady-states are densities, we see that the associated steady-state measure $\mu_{\infty}^{\lambda}([0,x])$  satisfies $\mu_{\infty}([1-y,1]) = \int_{1-y}^1 f_{\theta}(x) \dx$, and can find the Hölder exponent from the calculation 
\begin{align*} \ds\lim_{y \to 0} y^{-\alpha} \int_{1-y}^1 f_{\theta}(x) \dx &= \ds\lim_{y \to 0} y^{-\alpha} \ds\int_{1-y}^1  Z_f^{-1} x^{\lambda (\gamma - 1) - 2 \theta  - 1} \left(1 - x\right)^{\theta  - 1} \left(1 + x\right)^{-\lambda \gamma + \theta - 1} \dx\\  &= \ds\lim_{y \to 0}  \frac{(1-y)^{\lambda (\gamma - 1) - 2 \theta - 1} (y)^{ \theta  - 1} \left(2 - y\right)^{-\lambda \gamma + \theta - 1}}{Z_{f}  \alpha y^{\alpha - 1}} = \frac{2^{-\lambda \gamma + \theta - 1}}{\alpha Z_f} \ds\lim_{y \to 0} y^{\theta - \alpha}
\end{align*}
Then we see that \[ \lim_{y \to 0} y^{-\alpha} \mu_{\infty}^{\lambda}([1-y,1]) = \left\{
     \begin{array}{ll}
       0 & :   \alpha  < \theta \\
       (\alpha Z_f)^{-1} \:  2^{-\lambda \gamma + \theta - 1}& : \alpha = \theta \\
       \infty & : \alpha > \theta
     \end{array}
   \right.  \]
 which allows us to conclude that the \holder exponent of the steady-state $f_{\theta}(x)$ is $\theta_0$, and therefore the long-time \holder exponent  $\theta_{\infty} = \theta_0$, i.e. the \holder exponent of the steady-state distribution $\int_{0}^{x}f_{\theta}(x) \dx$ agrees with the \holder exponent $\theta_0$ of the initial data $\mu_0[0,x]$.
 Combined with the intuition provided in Example \ref{ex:exactsolution}, we further conjecture that the \holder exponent is conserved in time by the multilevel system given by Equation \ref{eq:pdeparam}. The discrepancy between the potentially nonzero \holder exponent for finite times in Example \ref{ex:exactsolution} and the \holder exponent of $0$ the weak-limit distribution $\delta(x)$ is reminscent of the discrepancy between conserved energy and the energy of weak-limit solutions for solutions of the Becker-Döring equations with supercritical initial density \citep{ball1986becker}. We demonstrate the preservation of Hölder exponents for the explicit solutions from Equation \ref{eq:PDspecialsolution} with our specially chosen initial data in Appendix \ref{sec:Holderpreserved}.
 
 \end{remark}

 \subsection{Steady States for Prisoner's Dilemma}\label{sec:PDsteadystates}
 
 We can find the possibly steady-state densities for the multilevel Prisoner's Dilemma system by characterizing time-independent solutions of Equation \ref{eq:specialPDPDE}. We seek to solve the equation 
 \begin{equation} \label{eq:timeindependentspecialPD} 0 = \dsdel{}{x} \left(x(1-x)(1+x) f(x) \right) + \lambda f(x) \left[ \gamma x - x^2 - \gamma M_1^f + 2 M_2^f \right]\end{equation}
Density solutions to Equation \ref{eq:timeindependentspecialPD} are steady-states of the form $$f(x) = Z_f^{-1} x^{\lambda \left(\gamma M_1^f - M_2^f \right) - 1} \left(1 - x\right)^{\left(\lambda / 2\right) \left( \gamma - 1 -  (\gamma M_1^f - M_2^f)\right)  - 1} \left(1 + x\right)^{-1 + \left(\lambda ./ 2 \right) \left(1 + \gamma +  (\gamma M_1^f - M_2^f)\right)}$$
As a matter of self-consistency, the steady-state distribution has to be integrable (and, in fact, normalized). We note that the integral of $f(x)$ will blow up near $x=1$ unless the exponent of $1-x$ exceeds $-1$. This occurs when  
\[ \frac{\lambda}{2} \left( \gamma - 1 - \left(\gamma M_1^f - M_2^f \right) \right) > 0 \Longrightarrow G(1) > \int_0^1 G(x) f(x) \dx \]
where we recall that $G(1) = \gamma - 1$ and that $\int_0^1 G(x) f(x) \dx = \gamma M_1^f - M_2^f$. Therefore it is impossible to have a density steady-state for this family of Prisoner's Dilemmas for which the average payoff of the population exceeds the average payoff of a all-cooperator group. 
We can parametrize these steady-states by their H{\"o}lder coefficent near $x = 1$, which we can compute using Definition \ref{def:Holderexponent}. Definition \ref{def:Holderexponent} tell us that the Hölder exponent of our steady-state density $f^{\lambda}_{\theta}(x)$ near the endpoint $x=1$ is equal to
 \begin{equation} \label{eq:steadyHolder} \theta = \frac{1}{2} \left[\lambda (\gamma - 1) - \lambda ( \gamma M_1^f - M_2^f) \right]. \end{equation} 
We can also rearrange this to express the average payoff in the whole population ($\int_0^1 G(x) f(x) \dx = \gamma M_1^f - M_2^f$) in terms of the Hölder exponent $\theta$, for given $\lambda$ and $\theta$, as follows
\begin{equation} \label{eq:averagepopulationpayoff}\int_0^1 G(x) f(x) \dx =  \gamma M_1^f - M_2^f =  \gamma - 1  - \frac{2 \theta}{\lambda} \end{equation} %
Rewriting our steady-states in terms of $\theta$, we have \begin{equation} \label{eq:pdsteadybetatheta} f^{\lambda}_{\theta}(x) = Z_f^{-1}  x^{ \lambda (\gamma - 1) - 2 \theta - 1} \left(1 - x\right)^{\theta - 1} \left(1 + x\right)^{- \lambda \gamma + \theta- 1}  \end{equation}
We notice that the solution of the time-independent problem of Equation \ref{eq:timeindependentspecialPD} given by Equation \ref{eq:pdsteadybetatheta} with given Hölder exponent $\theta$ near $x=1$ coincides with the density from Proposition \ref{prop:steadystatespecialpd} achieved as the long-time behavior of Equation \ref{eq:specialPDPDE} for initial data with corresponding Hölder exponent $\theta$. 

 So far, we have classified the long-time behavior of our system when $\lambda (\gamma - 1) < 2 \theta$, in which case we converge to full defection, and when $\lambda (\gamma - 1) > 2 \theta$, showing that we get coexistence of many compositions of cooperators and defectors within groups. Now we can study how this coexistence between coperation and defection fares in the limit as $\lambda \to \infty$, when between-group competition is infinitely stronger than within-group selection. Denoting the steady-state which attracts initial distributions with Hölder exponent $\theta$ near $x = 1$ for given $\lambda$ (satisfying $\lambda (\gamma - 1) > 2 \theta$), we have that \begin{equation} \label{eq:specialpdflambdatheta} f^{\lambda}_{\theta}(x) = Z_f^{-1} x^{\lambda (\gamma - 1) - 2 \theta - 1} (1-x)^{\theta -1} (1+x)^{-\frac{\lambda \gamma}{2} + \theta - 1}  \end{equation}
One way of quantifying our steady-state $f^{\theta}_{\lambda}(x)$ is to describe the most abudant (or modal) group type at steady-state $\hat{x}^{\theta}_{\lambda} := \argsup_{x \in [0,1]} f^{\lambda}_{\theta}(x)$. In the following proposition, we characterize the behavior of $\hat{x}^{\theta}_{\lambda}$ for various $\lambda$, and show that the behavior of $\hat{x}^{\theta}_{\lambda}$ depends on whether average group payoff $G(x)$ is maximized by all-cooperator groups ($x^*=1$, when $\gamma \geq 2$) or by type of group with a mix of cooperators and defectors ($x^* \in (0,1)$, when $\gamma < 2$).
 
 \begin{proposition}  Suppose $\theta  \geq 1$. If  $1 \leq \gamma \leq 2$, $\ds\lim_{\lambda \to \infty} \hat{x}_{\lambda} (f^{\lambda}_{\theta}) = \gamma - 1$. If  $\gamma > 2$, $ \ds\lim_{\lambda \to \infty} \hat{x}_{\lambda} (f^{\lambda}_{\theta} ) = 1$. 

\end{proposition}

\begin{proof}

We start by differentiating $f^{\lambda}_{\theta}$ and see that $$\dsddx{f^{\lambda}_{\theta}(x)}{x} = g(x) \left[ Z_f^{-1}  x^{\lambda \left( (\gamma - 1) - 2 \theta \right) - 2} \left(1 - x\right)^{\theta - 2} \left(1 + x\right)^{ - \lambda \gamma + \theta - 2} \right] $$ where $g(x)$ is given by $ g(x) = %
 \lambda( \gamma - 1) - 2 \theta - 1 + \left(  - \lambda \gamma \right) x +  (3 + \lambda) x^2$.
We note that $g(x)$ vanishes at the points \begin{equation} \label{eq:specialxhatlambdapm} x^{\lambda}_{\pm} = \frac{\lambda \gamma + \ds\sqrt{\left(\lambda \gamma \right)^2 - 4 \left(\lambda (\gamma - 1) - 2 \theta  - 1 \right) ( 3 + \lambda)}}{2 ( 3 + \lambda)} \end{equation} In the limit of large $\lambda$, that this simplifies to \begin{equation} \label{eq:specialxhatlambdapminfinity} x^{\infty}_{\pm} := \ds\lim_{\lambda \to \infty} x^{\lambda}_{\pm} = \frac{\gamma}{2} \pm \ds\sqrt{\frac{\gamma^2}{4} - (\gamma - 1) }  = \frac{\gamma}{2} \pm \ds\sqrt{\left(\frac{\gamma}{2} - 1\right)^2}\end{equation} Therefore the critical points of $g(x)$ depend on the relative value of $\frac{\gamma}{2}$ and $1$, and we see that 

\begin{itemize}
 
\item For $\gamma < 2$, 
then $ x^{\infty}_{\pm} = \frac{\gamma}{2} \pm \left(1 - \frac{\gamma}{2} \right)$, and therefore $x^{\infty}_{+} = 1$ and $ x^{\infty}_{-} = \gamma - 1$. %
Further, $\gamma - 1 < \frac{\gamma}{2} < 1$, so the point $ x^{\infty}_{-} < \argmax_{x \in [0,1]} G(x)$ so the interor critical point at steady-state has fewer cooperators than the (interior) type of group that maximizes the average payoff of group members.

\item For $\gamma > 2$, %
$x^{\infty}_{\pm} = \frac{\gamma}{2} \pm \left(\frac{\gamma}{2} - 1 \right)$ and then $x^{\infty}_{-} = 1$ and $x^{\infty}_{+} = \gamma - 1 > 1$. Thus the unique critical point for $g(x)$ is $1$ for $x \in [0,1]$.    

\end{itemize}
For $\lambda > \frac{2 \theta}{\gamma - 1}$, we have from \eqref{eq:pdsteadybetatheta} that $f^{\lambda}_{\theta} (0) = f^{\lambda}_{\theta}(1) = 0$. Because $0$ and $1$ are the only possible critical points of $f^{\lambda}_{\theta}(x)$ other than those of $g(x)$, we can deduce that $\ds\lim_{\lambda \to \infty} \hat{x}_{\lambda}(f^{\lambda}_{\theta}) = \gamma - 1$ for $1 \leq \gamma \leq 2$ and that $\ds\lim_{\lambda \to \infty} \hat{x}_{\lambda}(f^{\lambda}_{\theta}) = 1$ for $\gamma \geq 1$. 
\end{proof}

\begin{remark} \label{rem:payoffphasetransition} For our simplified parameters, $\gamma = 1$ corresponds to the point where $R = P$ and our underlying game is no longer a Prisoner's Dilemma. $\gamma = 2$ corresponds to $2R = T + S$ and is the cutoff between $\gamma < 2$ where average group payoff is maximized at a fraction of cooperators $x^* \in (0,1)$ and $\gamma \geq 2$ where average group payoff is maximized all-cooperator groups. 

\end{remark}

\begin{remark} \label{rem:parameterphasetransition} An important observation is that $\hat{x}_{\lambda \to \infty}(f^{\lambda}_{\theta}(x)) < \argmax_{x \in [0,1]} G(x)$ precisely when $\gamma < 2$, meaning that peak abundance in steady-state under weak within-group selection results in fewer cooperators than in a group with maximal average payoff precisely when average payoff of group members is maximized by groups with fewer than 100 percent cooperators. When $\gamma \geq 2$, we can think of defectors as those that are purely dominant at the within-group level and cooperators are those that are purely dominant at the between-group level, so relative ratios of within-group and between-group selection strength between $\lambda = 0$ and $\lambda = \infty$ neatly interpolate between $\hat{x}_0 = 0$ and $\hat{x}_{\lambda \to \infty} = 1$. When groups with a mix of cooperators and defectors can outperform pure cooperator groups in competition at the group level, defectors can also be (at least partially) selected for at the between-group level, and the presence of defectors at steady-state exceeds the level expected by looking at optimal production at the between-group level alone.
\end{remark}
\begin{remark} \label{rem:groupx1payoff}
Using Equation \ref{eq:averagepopulationpayoff}, we see that the average payoff of the population satisfies $$\ds\lim_{\lambda \to \infty} \ds\int_0^1 G(x) f^{\lambda}_{\theta}(x) \dx = \lim_{\lambda \to \infty} \gamma M_1^f - M_2^f = \gamma - 1 = G(1).$$ %
This means that the average payoff of the whole population %
at steady-state achieve the average payoff of a all-cooperator group when group selection strength $\lambda \to \infty$ (this is also true when $\gamma \geq 2$ because $\lim_{\lambda \to \infty} \hat{x}_{\lambda} = 1$%
). This means that the population cannot, on average, achieve higher payoff than an all-coopertor group, and illustrates the critical importance of all-cooperator groups for the success of the whole population. Further, this provides additional illustration of the impossibility to achieve cooperation at steady-state for any $\lambda$ when $\gamma = 1$, where $G(x) = x(1-x)$ and the only group type with the same average payoff as the all-cooperator group is the all-defector group. 
\end{remark}
We illustrate the results of this section in three figures below. In Figure \ref{fig:pdghost}, we show, for various $\gamma$, the fraction of cooperators for groups that achieve maximal average payoff and for groups that achieve peak abundance at steady-state ($\hat{x}_{\lambda}$) when the relative strength of group-level selection $\lambda \to \infty$. In Figure \ref{fig:lambdagammaheatmapspecial}, we plot $\hat{x}_{\lambda}$ for various levels of $\lambda$ and $\gamma$. In Figure \ref{fig:lambdapeak}, we plot the mean fraction of cooperators at steady-state $M_1^{f_{\theta}}$ and fraction of cooperators in most abundant group type $\hat{x}_{\lambda}$ for different relative intensities of selection $\lambda$, which highlights the discrepancy between most fit group type and most abundant group type as $\lambda to \infty$ when $\gamma < 2$ and groups are best off with a mix of cooperators and defectors.

\begin{figure}[H]   
  \centering
    \includegraphics[width=0.6\textwidth]{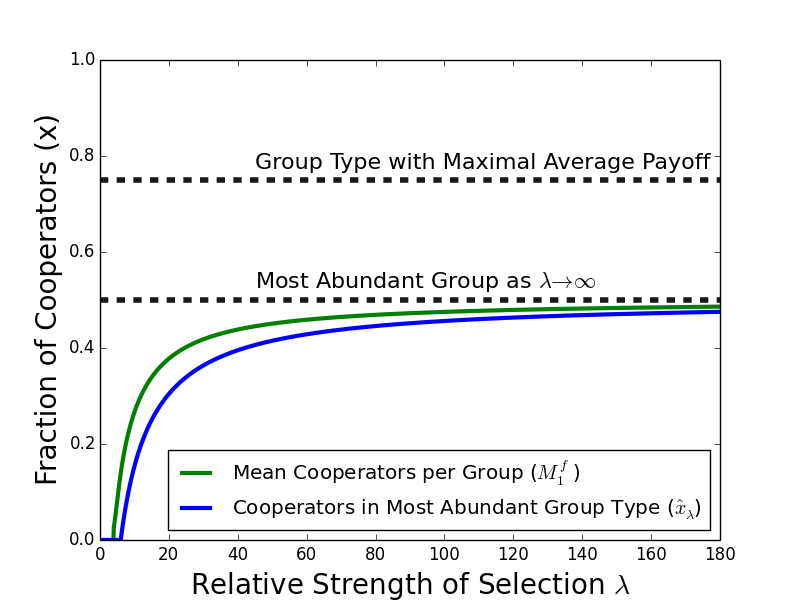}
      \caption{The mean fraction of cooperators $M_1^f$ and the fraction of cooperators in the most abudant group composition for the steady-state $\hat{x}_{\lambda}$ (as calculated in Equation \ref{eq:specialxhatlambdapminfinity}) for various values of  $\lambda$. Other parameters are fixed as $\gamma = \frac{3}{2}$, $\beta = 1$, $\alpha = 1$, so group average payoff is maximized at $x^* = \tfrac{3}{4}$ and peak abundance satisfies $\hat{x}_{\lambda} \to \frac{1}{2}$ as $\lambda \to \infty$. We note that the mean $M_1^{f_{\theta}}$ is initially larger than $\hat{x}_{\lambda}$ but also tends to $\tfrac{1}{2}$ as $\lambda \to \infty$. }
       \label{fig:lambdapeak}
    
\end{figure}

\begin{figure}[H]
  \centering
    \includegraphics[width=0.65\textwidth]{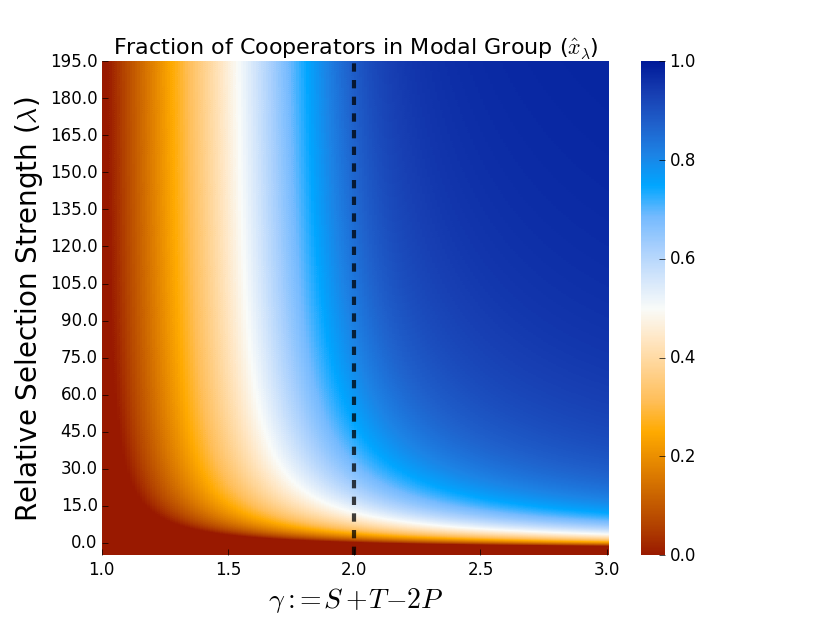}
      \caption{Color indicates the fraction of cooperators in the group type with peak abundance for the steady-state $\hat{x}_{\lambda}$ (as calculated in Equation \ref{eq:specialxhatlambdapm}) for various values of  $\lambda$ and $\gamma$. Other parameters are fixed as $\beta = 1$ and $\alpha = 1$, so group average payoff is maximized at $x^* = \min\left(\tfrac{\gamma}{2},1\right)$ and peak abundance satisfies $\hat{x}_{\lambda} \to \min\left(\gamma - 1,1\right)$ as $\lambda \to \infty$. Vertical slices of heatmap can be interpreted in the same way as the green curve for $\hat{x}_{\lambda}$ in Figure \ref{fig:lambdapeak} for a fixed value of $\gamma$.}
       \label{fig:lambdagammaheatmapspecial}
\end{figure}

\begin{figure}[H]
  \centering
    \includegraphics[width=0.6\textwidth]{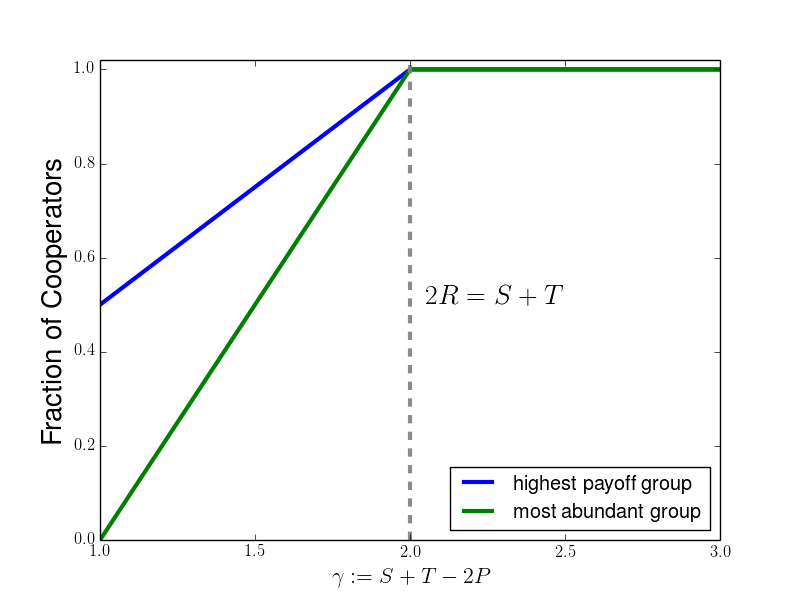}
      \caption{Comparison of the type of group composition $x$ with maximum group reproduction rate $G(x)$ to the peak abundance for the steady-state $\hat{x}_{\lambda}$ as $\lambda \to \infty$, plotted in terms of the parameter $\gamma = S + T - 2P$. For $\gamma \geq 2$ both peak group fitness and most abundant group type are all-cooperator groups, while for $\gamma < 2$, the most fit group type $\tfrac{\gamma}{2}$ exceeds the most abudant group type at steady-state $\gamma - 1$. }
       \label{fig:pdghost}
\end{figure}

\subsubsection{Luo-Mattingly Model} \label{subsec:lmmodel}

For Case II, corresponding to a scaled version of the frequency independent model governed by Equation \ref{eq:luomattingly}, it was shown by Luo and Mattingly that steady-state solutions take the form $f(x) = Z_f^{-1} x^{\lambda - \theta - 1} (1-x)^{\theta - 1} $ \citep{luo2017scaling} for $\lambda > \theta > 0$. They showed that these steady-states are achieved as the long-time behavior of Equation \ref{eq:luomattingly} for initial data with Hölder exponent $\theta$ near $x=1$. To find the most abundant group type at steady-state, we compute that $f'(x) = Z_f^{-1} \left[(2 - \lambda x) + (\lambda - \theta - 1) \right] x^{\lambda - \theta - 2} (1-x)^{\theta - 2}$. Then the three possible global maxima for $f(x)$ are the interior critical point $x = \frac{\lambda - \theta - 1}{\lambda - 2 } = 1 - \frac{\theta - 1}{\lambda - 2 }$ and the interval endpoints $x=0,1$. 
We then find that peak abundance is achieved by
\begin{displaymath}
   \ds\argmax_{x \in [0,1]} f(x) = \left\{
     \begin{array}{lr}
       0 & : \lambda < \theta + 1\\
       1 - \frac{\theta - 1}{\lambda - 2} & : \lambda > \theta + 1
     \end{array}
   \right.
\end{displaymath} 
and see that $\hat{x}_{\lambda} \to 1$ as $\lambda \to \infty$, and that $\hat{x}_{\lambda}$ interpolates between $0$ when $\lambda \leq \theta + 1$ to $1$ as $\lambda \to \infty$. Compared to Case I, we see that there is no discrepancy between the most fit group and the most abundant group type at steady-state in the limit of strong between-group selection ($\lambda \to \infty$), but there is also no possibility of having an intermediate group fitness optimum within this family of payoff matrices.

\section{Hawk-Dove} \label{sec:HD}

In this section, we consider the multilevel dynamics of the Hawk-Dove game. In Section \ref{sec:HDsolvable}, we consider a special family of payoff matrices which have solvable replicator dynamics, allowing us to apply the method of characteristics to characterize the long-time behavior of the system. In Section \ref{sec:HDsteadystate}, we %
analyze the steady-states of our multilevel PDE and consider the behavior of steady-states in the limit of large $\lambda$. 

\subsection{Hawk-Dove Game with Solvable Replicator Dynamics} \label{sec:HDsolvable}

We consider a special case for HD, where our payoffs satisfy $T = R + 1$ and $S = P + 1$, corresponding to the parameters $\beta = 1$, $\alpha = -2$. These games are characterized by payoff matrices of the form
\begin{equation} \label{eq:specialhdpayoffmatrix}
\begin{blockarray}{ccc}
& C & D \\
\begin{block}{c(cc)}
C & \gamma + P - 2 &    P + 1 \\
D & \gamma + P  - 1 & P \\
\end{block}
\end{blockarray}
\end{equation}
where, we need $\gamma > 3$ for such a payoff matrix to describe an HD game (so that $R > S$ and corespondingly $\gamma + P -2 > P +1$). The group average payoff function for this family of games is $G(x) = \gamma x - 2 x^2$, which is maximized by groups with a compostion of cooperators of $x^* = \min(\tfrac{\gamma}{4},1)$. In other words, groups are best off with $\frac{\gamma}{4}$ cooperators when $3 < \gamma < 4$, by all-cooperator groups when $\gamma \geq 4$. In the edge case when $\gamma = 3$ and the payoff matrix of Equation \ref{eq:specialhdpayoffmatrix} no longer corresponds to an HD game, groups have highest average payoff with $\tfrac{3}{4}$ cooperators.  %
From the payoff matrix, we have that our multilevel dynamics are given by 
\begin{equation} \label{eq:hawkdovePDE} \dsdel{f(t,x)}{t} = \dsdel{}{x} \left( x (1-x) (2x - 1) f(t,x)  \right) + \lambda f(t,x) \left[\left(\gamma x - 2x^2\right) - \left( \gamma M_1^f - 2M_2^f \right) \right], \end{equation}
which has characteristic curves satisfying the replicator dynamics \begin{equation} \label{eq:specialHDrep}\dsddt{x(t)} = x(1-x)(1-2x) \end{equation} which has a stable interior fixed point with a fifty-fifty mix of cooperators and defectors ($x = \frac{1}{2}$), and has unstable fixed points with full defector ($x=0$) and all-cooperator ($x=1$) groups. These replicator dynamics can be solved forwards in time with initial condition $x_0$ as \begin{equation}  \label{eq:hdcharacteristics}  x(t,x_0) = \phi_t(x_0)  =  \left\{
     \begin{array}{lr}
       \frac{1}{2} \left(1 + \frac{(2 x_0 - 1)}{\sqrt{(2x_0 - 1)^2 + (1 - (2x_0 - 1)^2) e^t}} \right) & : x_0 > \frac{1}{2}\\
       \frac{1}{2} \left(1 - \frac{(1 - 2 x_0)}{\sqrt{(1 - 2x_0)^2 + (1 - (1 - 2x_0)^2) e^t}} \right) & : x_0 < \frac{1}{2}
     \end{array}
   \right.  \end{equation} and can also be solved backwards in time for initial condition $x_0$ given $(t,x(t))$ as \begin{equation} \label{eq:hdcharbackwards} x_0(t,x(t)) = \phi_t^{-1}(x) = %
 \left\{
     \begin{array}{lr}
       \frac{1}{2} \left(1 + \frac{(2 x - 1)}{\sqrt{(2x - 1)^2 + (1 - (2x - 1)^2) e^{-t}}} \right) & : x > \frac{1}{2}\\
       \frac{1}{2} \left(1 - \frac{(1 - 2 x)}{\sqrt{(1 - 2x)^2 + (1 - (1 - 2x)^2) e^{-t}}} \right) & : x < \frac{1}{2}
     \end{array}
   \right. 
\end{equation} 

The component of our solution describing solely the effect of between-group competition is \begin{equation} \label{eq:wexpress} w_t(x) = \exp \left( \lambda \left( \gamma \ds\int_0^t x(s,x_0) ds - 2 \ds\int_0^t x(s,x_0)^2 ds \right) - \lambda \ds\int_0^t h(s) ds \right) \end{equation} where $h(s) := \gamma M_1^f(s) - 2 M_2^f(s)$. Using the integrals computed in the appendix, we have that 

 \begin{displaymath}
   w_t(x)  = \left\{
     \begin{array}{lr}
       \ds\frac{\left( \sqrt{(2x-1)^2 + (1 - (2x-1)^2)e^{-t}} + (2x -1)\right)^{\lambda ( \gamma - 2)} }{(2x)^{\lambda (\gamma - 2)} \left( (2x-1)^2 + (1 - (2x-1)^2)e^{-t} \right)^{ \frac{\lambda}{2}}} e^{\lambda ( \gamma - 2) t  - \lambda \int_0^t h(s) ds}  & : x >\frac{1}{2}\\
        \ds\frac{\left( \sqrt{(1-2x)^2 + (1 - (1-2x)^2)e^{-t}} + (1-2x)\right)^{\lambda ( 2 - \gamma)} }{(2x)^{\lambda (2 - \gamma)} \left( (1-2x)^2 + (1 - (1-2x)^2)e^{-t} \right)^{ \frac{\lambda}{2}}} e^{-\lambda \int_0^t h(s) ds} & : x < \frac{1}{2}
     \end{array}
   \right.
\end{displaymath}
For $a \in [0,\frac{1}{2}]$, we note that 
\begin{equation} \label{eq:aequation} \lambda (\gamma - 2)t - \int_0^t h(s) ds = \lambda (1-a) \left[ \gamma - 2 (1 + a) \right] - \lambda \gamma \ds\int_0^t \left(M_1^f(s) - a\right) ds + 2 \lambda 
\int_0^t \left( M_2^f(s) - a^2 \right)  ds \end{equation}
where %
$M_1^f(s) - a$ and $M_2^f(s) - a^2$ measure the deviation of the moments of our distribution $f(t,x)$ from the moments of $\delta(x - a)$. %
In our subsequent analysis, it is useful to denote  \[j_a(s):=  \gamma \left(M_1^f(s) - a \right)  - 2 \left(M_2^f(s) - a^2 \right). \]
We isolate the $x$-dependence by writing $w_t(x) = k(t) g_t(x) $ with $g_t(x)$ given by

 \begin{equation} \label{eq:hdgtx}
   g_t(x)  = \left\{
     \begin{array}{lr}
       \ds\frac{\left( \sqrt{(2x-1)^2 + (1 - (2x-1)^2)e^{-t}} + (2x -1)\right)^{\lambda ( \gamma - 2)} }{(2x)^{\lambda (\gamma - 2)} \left( (2x-1)^2 + (1 - (2x-1)^2)e^{-t} \right)^{ \frac{\lambda}{2}}}  & : x >\frac{1}{2}\\
        \ds\frac{\left( \sqrt{(1-2x)^2 + (1 - (1-2x)^2)e^{-t}} + (1-2x)\right)^{\lambda ( 2 - \gamma)} }{(2x)^{\lambda (2 - \gamma)} \left( (1-2x)^2 + (1 - (1-2x)^2)e^{-t} \right)^{ \frac{\lambda}{2}}}  & : x < \frac{1}{2}
     \end{array}
   \right.
\end{equation}
For the Hawk-Dove game, where there can be groups with fewer cooperators than at the within-group equilibrium, we define the \holder exponent describing the behavior near the $x = 0$ endpoint as $$\zeta = \ds\inf \left\{ \rho \: \bigg| \: \ds\lim_{x \to 0} \frac{\mu_0[0,x]}{x^{\rho}} > 0\right\} $$
To characterize the long-time behavior of the system, we use the following two lemmas. Lemma \ref{lem:lefthalf} shows us that groups with fewer cooperators than present in the within-group Hawk-Dove equilibrium at $x=\tfrac{1}{2}$ are not present in the long-run steady-state. We show that probability of having groups with fewer cooperators than at the interior Hawk-Dove equilibrium decays to $0$ exponentially quickly as $t \to \infty$.  %
Lemma \ref{lem:momentinequality} makes a technical point, telling us that%
$\int_0^{\infty} j_a(s) ds > -\infty$ for $a \in [0,\tfrac{1}{2})$, which allows us to simplify our analysis for characterizing when the distribution $f(t,x)$ can converge to a delta-function at the Hawk-Dove equilibrium $\delta(x - \tfrac{1}{2})$. 

For these lemmas, and for the subsequent characterization of the long-time behavior of Equation \ref{eq:hawkdovePDE}, we will assume that the initial distribution has positive Hölder exponent near $\zeta > 0$ near $x = 0$, so there is no initial mass concentrated at all-defector groups. We note that this assumption is made for convenience, as is allows us to make direct use of the method of characteristics. However, similar behavior should be expected with initial probably concentrated at all defector groups, as such groups should lose out to groups with initial nonzero fractions of cooperators due to between-group competition and groups with initially positive fractions of cooperators are expected to be attracted to the Hawk-Dove equilibrium at $x=\tfrac{1}{2}$ due to within-group competition.

\begin{lemma} \label{lem:lefthalf} Suppose we have an initial distribution $\mu_0$ with Hölder exponent $\zeta > 0$ near $x = 0$. Then for each $\delta > 0$,  $\ds\lim_{t \to \infty} \mu_t[0,\frac{1}{2} - \delta] = 0$.  \end{lemma}

\begin{proof}  From the measure-valued formulation%
, we write for continuous test-function $\psi(x)$ that $$\ds\int_0^{\frac{1}{2} - \delta} \psi(x) d \mu_t (x)  = \ds\int_0^{\phi_t^{-1}(\frac{1}{2} - \delta)} \psi (\phi_t(x)) w_t(\phi(x)) d \mu_0(x) \leq ||\psi||_{\infty} \ds\int_0^{\phi_t^{-1}(\frac{1}{2} - \delta)} w_t(\phi_t(x)) d\mu_0(x)$$ Because $g_t(x)$ is bounded on $[0,\tfrac{1}{2} - \delta]$, there is some $M < \infty$ such that $w_t(x) \leq M e^{ -\lambda \int_0^t h(s) ds}$, and thus  $$\ds\int_0^{\frac{1}{2} - \delta} \psi(x) d \mu_t (x) \leq M e^{  -\lambda \int_0^t h(s) ds} \mu_0[0,\phi_t^{-1}(\tfrac{1}{2} - \delta)] \leq M \mu_0[0,\phi_t^{-1}(\tfrac{1}{2} - \delta)] $$
where the last inequality follows because $h(s) := \gamma M_1^{\mu} - 3 M_2^{\mu} \geq 3(M_1^{\mu} - M_2^{\mu}) \geq 0$.
For sufficiently large $t$, we see from Equation \ref{eq:hdcharbackwards} and our assumption on $\mu_0[0,x]$, that \begin{align*} \mu_0[0,\phi_t^{-1}(x)] &\approx \frac{C}{2^{\zeta}}  \left( 1 - \frac{(1 - 2 x)}{\sqrt{(1 - 2x)^2 + (1 - (1 - 2x)^2) e^{-t}}} \right)^{\zeta}  \\ &=
 \frac{C}{2^{\zeta}} e^{-\zeta t} \left( \frac{4x(1-x) +  e^t \left((1 - 2x)^2  - (1-2x) \sqrt{(1 - 2x)^2 + (1 - (1-2x)^2) e^{-t}}\right)}{(1 - 2x)^2 + (1 - (1-2x)^2) e^{-t}} \right)^{\zeta}  \end{align*}
For the term in parenthesis, we see that the numerator is $\leq 4x(1-x)$ and the denominator is $\geq (1-2x)^2$, so we have that
$$ \mu_0[0,\phi_t^{-1}(x)] \leq \frac{C}{2^{\zeta}} e^{-\zeta t} \left(4x(1-x) \right)^{\zeta} \left(1 - 2x\right)^{- 2 \zeta} \leq \frac{C \delta^{-2 \zeta}}{2^{\zeta}} e^{-\zeta t}$$ 
and choosing the test function $\psi(x) \equiv 1$, we conclude that 
$$\mu_t[0,\tfrac{1}{2} - \delta] \leq \frac{CM \delta^{- 2 \zeta}}{2^{\zeta}} e^{-\zeta t} \to 0 \: \: \mathrm{as} \: \: t \to \infty \eqno{\qedhere}  $$
  \end{proof}

  \begin{lemma} \label{lem:momentinequality} If $\mu_0[0,x]$ has Hölder exponent $\zeta$ near $ x = 0$, then for $a \in [0,\tfrac{1}{2})$, $\int_0^{\infty} j_a(s) ds > - \infty$. %
  \end{lemma}
\begin{proof} We start by writing $$j_a(s)  = \int_0^a \left[\gamma \left( x - a \right) - 2\left( x^2 - a^2 \right)  \right]\mu_t(dx) + \int_a^1 \left[\gamma \left( x - a \right) - 2\left( x^2 - a^2 \right)  \right] \mu_t(dx)  := j_a^1(s) + j_a^2(s)$$ We denote the intervals $\mc{I}_1 = [0,a]$ and $\mc{I}_2 = (a,1]$ and correspondingly define $p_1 = \int_0^a \mu_t(dx)$ and $p_2 = \int_{a}^1 \mu_t(dx)$.

First we analyze $ \int_a^1 \left[ \gamma \left( x - a \right)  - 2 \left(x^2 - a^2\right) \left( \frac{1}{p_2} \right) \right] \mu_t(dx)$. Using the change of variable $z = \frac{x-a}{1-a}$, we can convert our description of $\frac{1}{p_2} \mu_t[0,x]$ as a probability measure on $\mc{I}_2 = [a,1]$ into a probability measure $\nu_t[0,z]$ with support on $[0,1]$. Because we require normalization when integrating with respect to $z$, we desire
$$1 = \ds\int_{a}^{1} \frac{1}{p_2} \mu_t(dx) = \ds\int_0^1 \frac{1-a}{p_2} \mu_t(dz) \Longrightarrow \nu_t(dz) := \frac{1-a}{p_2} \mu_t(dz) $$
Using this, we can now compute that \[ \ds\int_{a}^1 \frac{\left(x - a \right)}{p_2}  \mu_t(dx) = (1-a) \ds\int_0^1 z \nu_t(dz) = (1-a) M_1^{\nu} \] %
where $M_j^{\nu} := \int_0^1 z^j \nu_t(dz)$. Similarly, we find that
 \[\ds\int_{a}^1 \frac{\left(x^2 - a^2 \right)}{p_2} \mu_t(dx) =  \int_0^1 \left[\left((1 - a) z + a\right)^2 - a^2\right] \nu_t(dz) = (1-a) \left[2a M_1^{\nu} + (1-a) M_2^{\nu} \right]
\]
And we can further see that \begin{align*} j_a^2(s) &= p_2 (1-a) \left[ \gamma  \ds\int_{a}^1 \frac{\left(x - a \right)}{p_2}  \mu_t(dx) - 2  \ds\int_{a}^1 \frac{\left(x^2 - a^2 \right)}{p_2}  \mu_t(dx)  \right] \\ &= p_2 (1-a) \left[ \left( \gamma - 4a \right) M_1^{\nu} - 2 (1-a) M_2^{\nu} \right] \\ &= p_2 (1-a) \left( \gamma - 2(1+a) \right) M_1^{\nu} + 2(1-a) (M_1^{\nu} - M_2^{\nu}) \geq 0 \end{align*}
where we obtain the final inequality because $p_2 \geq 0$, $a < \frac{1}{2}$ by assumption, $\gamma \geq 3$ for the Hawk-Dove game and that $M_1^{\nu} \geq M_2^{\nu}$ because $\supp(\nu_t(dz)) \subseteq [0,1]$. 

Now we can estimate that\begin{align*} j_a^1 & = \ds\gamma \int_0^a \left(x - a \right) \mu_t(dx)  - 2 \underbrace{\int_0^a \left(x^2 - a^2 \right) \mu_t(dx)}_{\leq 0} \geq -\gamma a \int_0^a \mu_t(dx)   = - \gamma a \mu_t[0,a] \end{align*}
Because $a < \frac{1}{2}$, $\exists \delta_a > 0$ such that $a = \tfrac{1}{2} - \delta_{a}$. Using this and the result of Lemma \ref{lem:lefthalf}, we find that there are $C, M < \infty$ such that
\begin{align*} j_a^1(s) &\geq  - \gamma a \mu_t[0,\tfrac{1}{2} - \delta_a] \geq -  \frac{CM \gamma a \delta_a^{- 2 \zeta }}{2 ^{\zeta}} e^{- \zeta t} 
\end{align*} 
Combining this with our finding that $j_a^2(s) \geq 0$, we can see that $$\ds\int_0^t j_a(s) ds = \int_0^t \left( j_a^1(s) + j_a^2(s) \right) ds \geq  - \ds\int_0^t  \left(  \frac{CM \gamma a \delta_a^{- 2 \eta }}{2 ^{\zeta }} e^{- \zeta s}  \right) ds =   \frac{CM \gamma a \delta_a^{- 2 \zeta}}{\zeta 2 ^{\zeta }} \left(  1 - e^{-\zeta t}\right)  $$
and therefore we see that $\ds\int_0^{\infty} j_a(s) ds \geq -  \frac{CM \gamma a \delta_a^{- 2 \zeta}}{\zeta 2 ^{\zeta}} > - \infty$.  
\end{proof}

Now we are ready to characterize the long-time behavior of our multilevel Hawk-Dove system in terms of its initial distribution. The qualitative behavior is divided into two regimes by a critical level of relative strength of between-group selection ($\lambda^* := \frac{2 \theta}{\gamma - 3}$) which depends on the payoff matrix through $\gamma$ and on the initial data through the \holder exponent $\theta$ near $x=1$. For $\lambda < \lambda^*$, the distribution of groups converge to a delta-function at the within-group HD equilibrium ($\delta(x - \tfrac{1}{2})$) as $t \to \infty$. When $\lambda > \lambda^*$, the distribution of groups converges as $t \to \infty$ to a density which is supported at all group types between the within-group HD equilibrium at $x = \tfrac{1}{2}$ and full-cooperation at $x=1$. Notably, between-group selection can only increase the fraction of cooperators sustained in the group-structured population at steady-state.

\begin{proposition} \label{prop:deltahd} Suppose $\mu_0[0,x]$ has Hölder exponents $\zeta$ near $x=0$ and $\theta$ near $x=1$. Then if $\lambda (\gamma - 3) < 2 \theta$, $\mu_t(dx) \rightharpoonup \delta(x - \frac{1}{2})$. %
 \end{proposition}

\begin{proof}

For $\lambda (\gamma - 3) < 2 \theta$, 
we show that, for any continuous function $\psi(x)$, $\int_0^1 \psi(x) \mu_t(dx) \to \psi(\tfrac{1}{2})$ as $t \to \infty$. Because $\phi(\cdot)$ is continuous, we know that $\forall \epsilon > 0$, $\exists \delta > 0$ such that $|\phi(x) - \phi(\tfrac{1}{2})| < \epsilon$ when $|x - \tfrac{1}{2}| < \delta$. Because $\mu_t(dx)$ is a probability distribution, we can say that \begin{align*}\bigg| \ds\int_0^1 \psi(x) \mu_t(dx) - \psi(\tfrac{1}{2}) \bigg| %
&\leq \int_0^{\tfrac{1}{2} - \delta} |\psi(x) - \psi(\tfrac{1}{2}) | \mu_t(dx) +  \int_{\tfrac{1}{2} - \delta}^{\tfrac{1}{2} + \delta} |\psi(x) - \psi(\tfrac{1}{2}) | \mu_t(dx)+ \ds\int_{\tfrac{1}{2} + \delta}^1 |\psi(x) - \psi(\tfrac{1}{2})| \mu_t(dx)  \\ &\leq \epsilon + 2 ||\psi||_{\infty} \int_{0}^{\phi_t^{-1}(\tfrac{1}{2} - \delta)} w_t(\phi_t(x)) \mu_0(dx) %
+ 2 ||\psi||_{\infty} \ds\int_{\phi_t^{-1}(\tfrac{1}{2} + \delta)}^1  w_t(\phi_t(x)) \mu_0(dx) \end{align*}
We have that $\exists M > 0$ such that $g_t(\phi_t(x)) \leq M$ for each $x \in [0, \phi^{-1}_t(\tfrac{1}{2} - \delta)] \cup [\phi^{-1}_t(\tfrac{1}{2} + \delta),1]$, so  \begin{align*}\bigg| \ds\int_0^1 \psi(x) \mu_t(dx) - \psi\left(\tfrac{1}{2}\right) \bigg| &\leq \epsilon + %
2  M ||\psi||_{\infty} e^{-\lambda \int_0^t h(s) ds} \mu_t[0,\tfrac{1}{2} - \delta] + 2 M ||\psi||_{\infty} e^{\lambda (\tfrac{\gamma - 3}{2}) t - \int_0^t h(s) ds} \mu_0[\phi_t^{-1}(\tfrac{1}{2} + \delta),1]%
  \end{align*}
  From Lemma \ref{lem:lefthalf}, %
we know that $2 ||\psi||_{\infty} e^{-\lambda \int_0^t h(s) ds} \mu_t[0,\tfrac{1}{2} - \delta] \to 0$ as $t \to \infty$. %
Now we examine $x \in [\tfrac{1}{2} + \delta,1]$. We see from Equation \ref{eq:hdcharacteristics} that \begin{align*} \phi_t^{-1}(x) = \frac{1}{2} \left( 1 +\frac{1}{\sqrt{1 + \left(\frac{1}{(2x-1)^2} - 1 \right)e^{-t}}} \right) \geq 1 - \frac{1}{2} \left( \frac{\left(\frac{1}{(2x-1)^2} - 1 \right)e^{-t}}{1 + \left(\frac{1}{(2x-1)^2} - 1 \right)e^{-t}} \right) \geq 1 - \frac{1}{2} \left(\frac{1}{(2x-1)^2} - 1\right) e^{-t} \end{align*}  so there exists $D>0$ such that for $x \in [\frac{1}{2} + \delta, 1]$, $$\mu_0[\phi_t^{-1}(x),1] \leq  \mu_0\left[1 - De^{-t},1 \right] \leq C D^{\theta} e^{-\theta t}$$ %
for sufficiently large $t$ %
due to our assumption on %
$\mu_0[1-y,1]$ for $y$ near $0$. Thus we can write that
\begin{equation} \label{eq:usefulinequality}  2 M ||\psi||_{\infty} e^{\lambda (\gamma - 2) t - \int_0^t h(s) ds} \mu_0[\phi_t^{-1}(\tfrac{1}{2} + \delta),1]  \leq 2 M C D^{\theta} ||\psi||_{\infty} e^{[\lambda (\gamma - 2) - \theta] t - \int_0^t h(s) ds } \end{equation}
Because $\tfrac{\lambda}{2} \left( \gamma - 3 \right) < \lambda (\gamma -2)$, we can separate our analysis into cases where $\lambda (\gamma - 2) < \theta$ and where $ \tfrac{\lambda}{2} \left( \gamma - 3 \right) < \theta < \lambda (\gamma - 2)$. In the case where $\lambda (\gamma - 2) < \theta$, then we use the fact that $h(s) := \gamma M_1^{\mu} - 2 M_2^{\mu} \geq 0$ (because $\gamma \geq 3$ for the Hawk-Dove game and because $\supp\left(\mu_t(dx)\right) \subseteq [0,1]$) and Equation \ref{eq:usefulinequality} to see that 
\begin{equation} \label{eq:righthalftozero} 2 M ||\psi||_{\infty} e^{\lambda (\gamma - 2) t - \int_0^t h(s) ds} \mu_0[\phi_t^{-1}(\tfrac{1}{2} + \delta),1]  \leq 2 M C D^{\theta} ||\psi||_{\infty} e^{[\lambda (\gamma - 2) - \theta] t } \to 0 \: \: \mathrm{as} \: \: t \to \infty   \end{equation}%
If $\tfrac{\lambda}{2} \left( \gamma -3 \right) < \theta < \lambda \left( \gamma - 2 \right)$, we use Equation \ref{eq:aequation}. Because $H(a) := (1-a)\left[\gamma - 2 (1+a) \right]$ is a continuous and decreasing function of $a$, our assumption on $\lambda$ tells us that $\exists a^* \in (0,\tfrac{1}{2})$ such that $\lambda H(a^*) = \theta$ and $\lambda H(a) < \theta$ for $a \in (a^*,\tfrac{1}{2}]$. Therefore, for $a^{**} \in (a^*,\tfrac{1}{2}]$ we can rewrite Equation \ref{eq:usefulinequality} as
\[ 2 M ||\psi||_{\infty} e^{\lambda (\gamma - 2) t - \int_0^t h(s) ds} \mu_0[\phi_t^{-1}(\tfrac{1}{2} + \delta),1]  \leq 2 M C D^{\theta} ||\psi||_{\infty} e^{[\lambda H(a^{**}) - \theta] t - \int_0^t j_{a^{**}}(s) ds }  \]
From Lemma \ref{lem:momentinequality}, we know that $\exists A < \infty$ such that $\forall t \geq 0$,  $\int_0^t j_{a^{**}}(s) ds \geq - A$, so, combined with our choice of $a^{**}$ so that $\lambda H(a^{**}) < \theta$, we can deduce that 
\[ 2 M ||\psi||_{\infty} e^{\lambda (\gamma - 2) t - \int_0^t h(s) ds} \mu_0[\phi_t^{-1}(\tfrac{1}{2} + \delta),1]  \leq 2 M C D^{\theta} ||\psi||_{\infty} e^{A}  e^{[\lambda H(a^{**}) - \theta] t } \to 0 \: \: \mathrm{as} \: \: t \to \infty \]
Knowing that $2 M ||\psi||_{\infty} e^{\lambda (\gamma - 2) t - \int_0^t h(s) ds} \mu_0[\phi_t^{-1}(\tfrac{1}{2} + \delta),1] \to 0$ as $ t \to \infty$ when $\lambda (\gamma - 3) < 2 \theta$, we can combine this with our deduction for $\mu[0,\tfrac{1}{2} - \delta]$ from Lemma \ref{lem:lefthalf} to conclude that %
$$\bigg| \ds\int_0^1 \psi(x) \mu_t(dx) - \psi(\tfrac{1}{2}) \bigg| <  \epsilon \: \: \mathrm{as} \: \: t \to \infty $$ and it follows that $\mu_t(dx) \rightharpoonup \delta(x - \tfrac{1}{2})$ as $t \to \infty$ when $\lambda (\gamma - 3) < 2 \theta$. \end{proof}

\begin{proposition} \label{prop:steadystatespecialhd} Suppose $\mu_0[0,x]$ has Hölder exponents $\zeta$ near $x=0$ and  $\theta$ near $x=1$. %
If $\lambda (\gamma - 3) > 2 \theta$, then  \begin{displaymath}
   \mu_t(dx) \to f_{\infty}(x) \dx   = \left\{
     \begin{array}{lr}
     Z_f^{-1} \left(2x -1 \right)^{\lambda (\gamma - 3) - 2 \theta - 1} (1-x)^{\theta - 1} x^{-\lambda(\gamma-2) + \theta - 1} \dx  & : x \geq \frac{1}{2}\\
       0 & : x < \frac{1}{2}
     \end{array}
   \right.
\end{displaymath}  %
 \end{proposition}

Sample steady-states for given initial distributions with $\theta = 2$ and various values of $\lambda$ are given in Figure \ref{fig:hdghostdensity} for payoff matrices with $\gamma = 2.5$ (left) and for $\gamma = 1.5$ (right). In both panels, we see that the fraction of cooperators  at steady-state increases as $\lambda$ increases, and we see from the right panel that even large relative selection strength of between-group selection $\lambda$ cannot sustain a substantial density of groups with optimal average payoff.

\begin{figure}[H]
    \centering
    \includegraphics[width=0.495\textwidth]{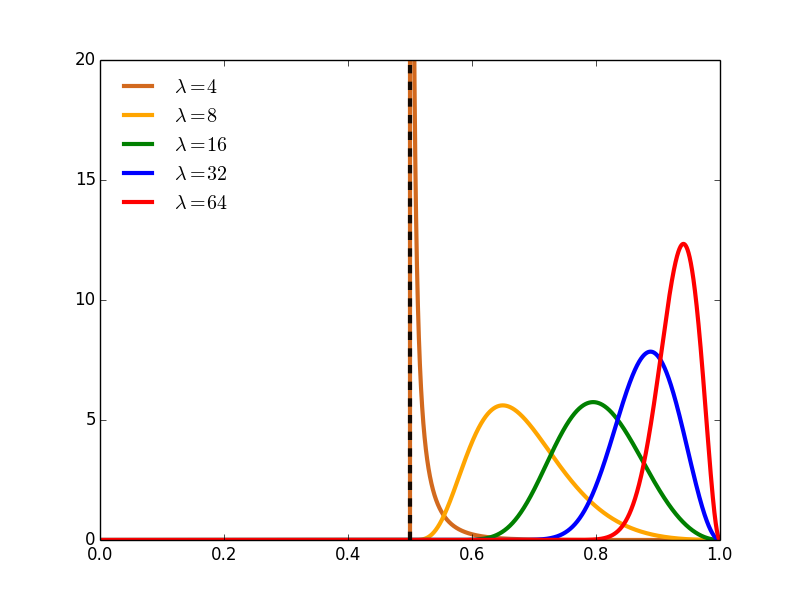}
    \includegraphics[width=0.495\textwidth]{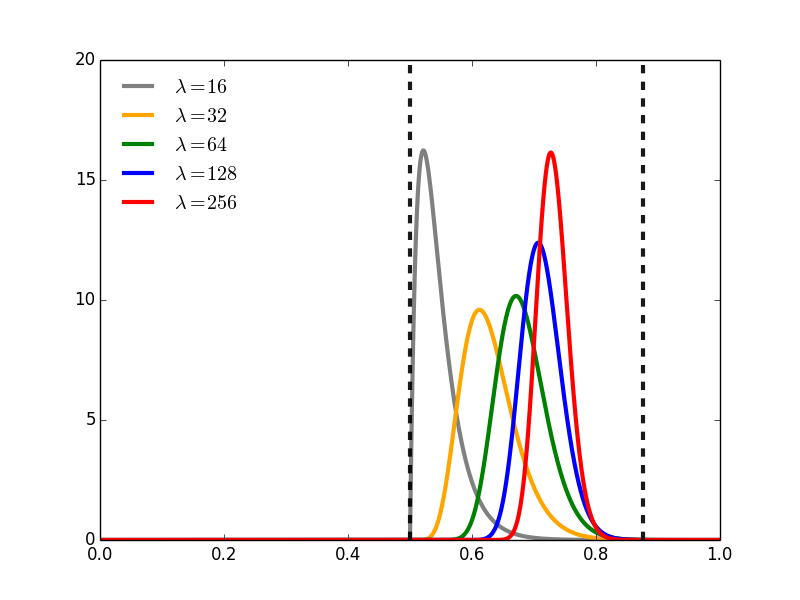}
    \caption{Steady state densities for various relative selection strengths $\lambda$ when initial distribution has Hölder exponent $\theta = 3$ near $x=1$, computed from the result of Proposition \ref{prop:steadystatespecialhd}. Dotted lines at $x = \tfrac{1}{2}$ indicates location of within-group Hawk-Dove equilibrium. (Left) Steady states for games with $\gamma = 4.5$ ($G(x)$ maximized by $x=1$).  (Right) Steady states for games with $\gamma = 3.5$ ($G(x)$ maximized by $x=\tfrac{7}{8}$). Dotted line at $x = \frac{7}{8}$ in right panel indicates maximal group payoff $G(x)$ at that composition $x$.}
    \label{fig:hdghostdensity}
\end{figure}

\begin{proof}
For $\lambda ( \gamma - 3) > 2 \theta$, we demonstrate, for some family of normalizing constants $\{Z_t\}_{t \geq 0}$%
, that $Z_t \int_{\frac{1}{2}}^1 \left[(\psi g_t) \circ \phi_t(x)\right]  \mu_t(dx) \to \int_{\frac{1}{2}}^1 \psi(x) f_{\theta}(x) \dx$. We first integrate by parts, obtaining $$\ds\int_{\frac{1}{2}}^1 [(\psi g_t) \circ \phi_t] (x) d \mu_0(x)  = \psi g_t F_0(1^-) - \psi g_t F_0(\tfrac{1}{2}^+) - \ds\int_{\frac{1}{2}}^1 \partial_x \left[(\psi g_t) \circ \phi_t  \right] (x) F_0(x) \dx  $$ where $F_0(x) :=  \mu_0[0,x]$ is the cumulative distribution function of $\mu_0$. Because $\phi_t(x)$ is continuous and satisfies $\phi_t(\frac{1}{2}) = \frac{1}{2}$ and $\phi_t(1) = 1$, we can also say that %
$\int_0^1 \partial_x \left[(\psi g_t) \circ \phi_t(x)  \right] \dx = [\psi g_t](1^-) - [\psi g_t ](\tfrac{1}{2}^+) $. This allows us to rewrite our above equation as  $$\ds\int_0^1 [(\psi g_t) \circ \phi_t] (x) d \mu_0(x)  =  \psi g_t(1 -  F_0(1^-)) - \psi g_t (1 - F_0(\tfrac{1}{2}^+)) + \ds\int_0^1 \partial_x \left[( \psi g_t) \circ \phi_t  \right] (x) (1 - F(x)) \dx  $$ 
We see that $1 - F_0(1^-) = 0$ because $\supp(\mu_t(dx)) \subseteq [0,1]$ and we can use Equation \ref{eq:hdgtx} to find that $g_t(\tfrac{1}{2}) = e^{-\tfrac{ \lambda (\gamma - 3)}{2}}$. Choosing $Z_t = e^{\theta t}$, we find that $$\ds\lim_{t \to \infty} e^{\theta t} [\psi g_t(1 - F_0)](\tfrac{1}{2}^+) = \ds\lim_{t \to \infty} \left( 1 - F_0(\tfrac{1}{2}) \right) \psi (\tfrac{1}{2}) e^{- (\frac{\lambda (\gamma - 3)}{2} - \theta) t} \to 0 \: \: \mathrm{when} \: \: \lambda (\gamma - 3) > 2 \theta $$
For the integral term, we write $y = \phi_t(x)$ and see that $\partial_x[\psi(y) g_t(y)] = \partial_y [\psi(y) g_t(y))] \cdot \partial_x y$, and integrating with respect to $y$ gives \begin{align*}  \ds\int_{\frac{1}{2}}^1 \partial_x \left[(\psi g_t) \circ \phi_t(x)  \right] (1 - F_0(x)) \dx   %
 &=   \ds\int_{\frac{1}{2}}^1 \left[\partial_y \left(\psi g_t\right) \circ y \right]  ( 1 - F_0(\phi_t^{-1}(y)) \dy \end{align*}
Because $g_t(x) \to (2x - 1)^{\lambda (\gamma - 3)} x^{- \lambda (\gamma - 2)}$, $\lim_{t \to \infty}\partial_x \left[\left(\psi g_t\right) \circ x \right]  = \partial_x [ \, (2x - 1)^{\lambda (\gamma - 3)} x^{- \lambda (\gamma - 2)} \psi(x) \,] $.
We also observe that \begin{align*} 1 - F_0(\phi_t^{-1}(x)) &= %
= \mu_0\left[1 -  \frac{1}{2} \left(1 - 
\frac{2 x - 1}{\sqrt{(2 x - 1)^2 + (1 - (2x - 1)^2) e^{-t}}} \right)   ,1 \right] \end{align*} Then, using our assumption on $\mu_0[1-x,1]$ for $x$ near $1$ to deduce that for large $t$, \begin{align*} 1 - F_0(\phi_t^{-1}(x)) &\approx %
\frac{C}{2^{\theta}} e^{-\theta t}  \left( \frac{  4x(1-x) + e^{t} \left((2x - 1)^2 - (2x - 1) \sqrt{(2x - 1)^2 + (1 - (2x-1)^2)e^{-t}}\right)}{(2x - 1)^2 + (1 - (2x-1)^2)e^{-t}} \right)^{\theta}   \end{align*} 
We can compute that %
$\ds\lim_{t \to \infty} e^{ \theta t} \left[1 - F_0( \phi_t^{-1}(x))\right]  = C 2^{1 - \theta} \left( x(1-x) \right)^{\theta} (2x-1)^{-2 \theta}$, and 
$$ e^{\theta t} \ds\int_{\frac{1}{2}}^1 \left[(\psi g_t) \circ \phi_t(x) \right] \mu_0(dx) \to \frac{C}{2^{\theta-1}} \ds\int_0^1 \partial_x \left[ (2x - 1)^{\lambda (\gamma - 3)} x^{- \lambda (\gamma - 2)} \psi(x) \right]  \left(\frac{x(1-x)}{(2x-1)^2} \right)^{\theta} \dx $$ and we can integrate the righthand side by parts to see that $$e^{ \theta t} \ds\int_0^1 \left[(\psi g_t) \circ \phi_t(x) \right] \mu_0(dx) \to \frac{\theta C}{2^{\theta-1}} \ds\int_0^1  (2x-1)^{\lambda (\gamma - 3) - 2 \theta - 1} (1-x)^{\theta - 1} x^{-\lambda (\gamma - 2) + \theta  - 1} \psi(x) \dx $$
Using Lemma \ref{lem:lefthalf}, we know that $\ds\lim_{t \to \infty} \mu_t[0,x) \to 0 $ for $x \leq  \frac{1}{2}$, and we can deduce that $$\mu_t(dx) \rightharpoonup \left\{
     \begin{array}{lr}
         Z_f^{-1} (2x-1)^{\lambda (\gamma - 3) - 2 \theta - 1} (1-x)^{\theta - 1} x^{-\lambda (\gamma - 2) + \theta  - 1} \dx & : x  \geq \frac{1}{2}\\
       0 & : x < \frac{1}{2}
     \end{array}
   \right.   \eqno{\qedhere}  $$
\end{proof}

\subsection{Steady States for Hawk-Dove Game} \label{sec:HDsteadystate}

We can find the possible steady-state behaviors for the multilevel Hawk Dove system by characterizing time-independent solutions of Equation \ref{eq:hawkdovePDE}. We seek to solve the equation
\begin{equation} \label{eq:reducedhd}0 = \dsdel{}{x} \left(x(1-x)(2x - 1) f(x) \right) +  \lambda f(x) \left[ \gamma x - 2 x^2 - \gamma M_1^f + 2 M_2^f \right]  \end{equation}

Due to the result of Lemma \ref{lem:lefthalf}, we expect no groups with fewer cooperators than in the within-group HD equilibrium at $x = \tfrac{1}{2}$, so we look for for steady-state densities of the form  %
 $$f(x) = \left\{
     \begin{array}{cr}
       0 & :  0 \leq x \leq \frac{1}{2} \\
       p(x) & : \frac{1}{2} < x \leq 1
     \end{array}
   \right.$$
   In other words, we look for steady-state densities with support between groups at the HD equilibrium and all-cooperator groups. Satisfactory densities $p(x)$ are of the form 
{ \large \[ \hspace{-10mm} p(x) = Z_f^{-1} x^{-\lambda \left(\gamma M_1^f - 2 M_2^f \right) - 1} \left( 1 - x \right)^{[\lambda (\gamma - 2)  - \lambda (\gamma M_1^f - 2 M_2^f)]-1} \left( 2x - \beta \right)^{ \left( \lambda ( 1 - \gamma) +2\lambda (\gamma M_1^f - 2 M_2^f) \right)-1} \] } As a test of self-consistency for our steady-state formula, in order for this density to actually correspond to a probability distribution, we require that the density is integrable. As a result, we need the power of $1-x$ to be greater than $-1$, which is satisfied when 
\[ \lambda (\gamma -2) -  \lambda (\gamma M_1^f - 2 M_2^f)  > 0 \Longrightarrow G(1) > \ds\int_0^1 G(x) f_{\infty}(x) \dx \]
where we recall that $G((1) = \gamma - 2$ and that $\int_0^1 G(x) f_{\infty}(x) \dx = \gamma M_1^f - 2 M_2^f$. Therefore it is impossible to have a density steady-state for this Hawk-Dove game for which the average payoff of the population $ \int_0^1 G(x) f_{\infty}(x) \dx$ exceeds the average payoff for members in a all-cooperator group $G(1)$. Because there are Hawk-Dove games for which $G(x)$ can exceed $G(1)$ (namely, when $3 < \gamma < 4$), this means that the long-time steady-state cannot fully take advantage of the optimal payoff achieved by groups with average payoff exceeding all-cooperator groups. To explore this, we aim to compute $\hat{x}_{\lambda} = \argsup_{x \in [0,1]}\{(f_{\theta}(x)\}$, the modal group type in the steady-state distribution.
We can parametrize our steady-states by the Hölder exponent $\theta$ of $p(x)$ near $x=1$, which is related to the exponent of the $(1-x)$ term by $$\theta = (\gamma - 2) \lambda - \lambda (\gamma M_1^f - 2 M_2^f)\Longrightarrow  \lambda \left(\gamma M_1^f - 2 M_2^f \right) = (\gamma - 2) \lambda -  \theta $$ %
allowing us to rewrite our steady-states as {\large \[ p_{\theta}(x)  = Z_f^{-1} x^{(2 - \gamma) \lambda + \theta  - 1} \left(1-x\right)^{\theta -1} \left( 2x - 1 \right)^{\lambda (\gamma -3) - 2 \theta - 1} \] }
We observe that the only steady-state of the multilevel HD PDE with given Hölder exponent $\theta$ near $x=1$ agrees with the steady-state shown in Proposition \ref{prop:steadystatespecialhd} to be achieved as the long-time behavior of Equation \ref{eq:hawkdovePDE} for initial distributions with the corresponding Hölder exponent. We also note that such a steady-state is integrable if and only if $\lambda (\gamma - 3) > 2 \theta $, and that if $\gamma = 3$, it is not possible to have an integrable steady-state of this form, regardless of the value of $\lambda$. 
 To determine the fraction of cooperators $x$ at which the density $p_{\theta}(x)$ is maximized, we compute $$p'_{\theta}(x) = Z_f^{-1} g(x) x^{ \left((2 - \gamma) \lambda + \theta \right) - 2} \left(1-x\right)^{\theta -2} \left( 2x - 1\right)^{ \left(\lambda (\gamma - 3) - 2 \theta \right) - 2} $$ where $$g(x) = (\gamma -2) \lambda - \theta + 1 - \left[\lambda \gamma + 4 \right] x+ \left[ 2 \lambda + 6 \theta \right] x^2$$ For sufficiently large $\lambda$ and $\theta$, we know from the above proposition that $p_{\theta}(\tfrac{\beta}{2}) = p_{\theta}(1) = 0$, so $f_{\theta}(x)$ is maximized at a root of $g(x)$. These roots are given by $$x^{\lambda}_{\pm} = \frac{\lambda \gamma + 4}{4 \lambda + 12} \pm \sqrt{\frac{\left( \lambda \gamma + 4\right)^2}{4 (2 \lambda + 6)^2} - \frac{(\gamma - 2) \lambda -  \theta + 1}{2 \lambda + 6}}$$
For large $\lambda$, we see that \begin{equation} \label{eq:hdpeakinfinity} x^{\infty}_{\pm} := \ds\lim_{\lambda \to \infty} x^{\lambda}_{\pm} = \frac{\gamma}{4} \pm \sqrt{\frac{\gamma^2}{16} - \frac{\left( \gamma - 2\right) }{2}}  = \frac{\gamma}{4} \pm \sqrt{\left( \frac{\gamma}{4} - 1\right)^2} \end{equation} We now break our analysis into two cases: 
\begin{itemize}
\item If $\gamma > 4$, then $x^{\infty}_{\pm} = \frac{\gamma}{4} \pm \left( \frac{\gamma}{4} - 1\right)$, so $x^{\infty}_+ = \frac{\gamma}{2} - 1 > 1$ and $x^{\infty}_{-} = 1$, so the most abundant group type in the steady-state $f_{\theta}(x)$ has cooperator fraction $x$ approach $1$ as $\lambda \to \infty$. 
\item If $\gamma \leq 4$, then $x^{\infty}_{\pm} = \frac{\gamma}{4} \pm \left( 1 - \frac{\gamma}{4} \right)$, so $x^{\infty}_{+} = 1$ and $x^{\infty}_{-} = \frac{\gamma}{2} - 1 \in (0,1)$. In this case, the most abundant group type in steady-state approaches an intermediate fraction $\frac{\gamma}{2} - 1 $ as $\lambda \to \infty$. 
\end{itemize}
In the extreme case of the HD game where $\gamma = 3$, Equation \ref{eq:hdpeakinfinity} yields $x^{\infty}_{-} |_{\gamma = 3} = \frac{3}{2} - 1 = \frac{1}{2}$. This tells us that as $\gamma$ approaches a value at which the density $p(x)$ cannot be integrable, %
then the most abundant group type has cooperator composition approaching $x^{\infty}_{-} |_{\gamma = 3}  = \frac{1}{2}$ even though the most fit group has $x^* = \tfrac{3}{4}$ cooperators, consistent with the prediction from Proposition \ref{prop:deltahd} for convergence to the within-group HD equilibrium $\delta(x - \frac{1}{2})$ when $\gamma = 3$.

In Figure \ref{fig:hdghost}, we illustrate the discrepancy between the most fit group type $x^*$ and the most abundant group type at steady-state $\hat{x}_{\lambda}$ as $\lambda \to \infty$. The figure shows how fewer cooperators are present than optimal for group reproduction $G(x)$ at steady-state even in the infinite $\lambda$ limit when groups are best off with an intermediate fraction of cooperators $x^* \in (0,1)$.

\begin{figure}[H] 
  \centering
    \includegraphics[width=0.6\textwidth]{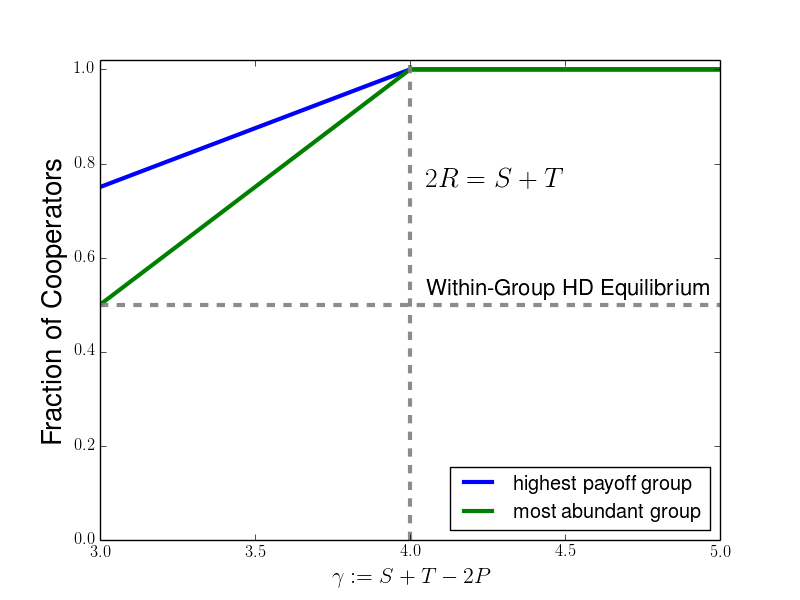}
      \caption{Comparison of the type of group composition $x$ with maximum group reproduction rate $G(x)$ to the peak abundance for the steady-state $\hat{x}_{\lambda}$ as $\lambda \to \infty$, plotted in terms of the parameter $\gamma = S+T-2P$. For $\gamma \geq 4$ both peak group fitness and most abundant group type are all-cooperatorx groups, while for $\gamma < 4$, the most fit group type $\tfrac{\gamma}{4}$ exceeds the most abundant group type at steady-state $\tfrac{\gamma}{2} - 1$. }
      \label{fig:hdghost}
\end{figure}

\addtocontents{toc}{\protect\setcounter{tocdepth}{1}}

\section{Discussion}

In this paper, we considered a model of multilevel selection in which within-group and between-group competition depend on payoffs from either the Prisoner's Dilemma (PD) or the Hawk-Dove (HD) game.
We studied a continuous-time stochastic process with explicit individual and group-level replication events, and derived a nonlocal PDE description of the multilevel system in the limit of large population size. We studied the steady-state behavior of this PDE, and characterized the critical ratio of relative selection strength $\lambda$ between the two levels which separate a regime in which defection dominates in all groups from a regime in which cooperation and defection can coexist in every group. 

\myindent In particular, we saw a qualitatively different behavior in two different scenarios: when all-cooperator groups have maximal average payoff (when ``Many hands make light work'') and when average payoff is maximized by groups with a mix of cooperators and defectors (``Too many cooks spoil the broth''). In the former case, we showed that defectors are most abundant at steady-state in the Prisoner's Dilemma 
when selection is strongest at the within-group level $(\lambda \to 0)$, while groups with many cooperators are most abundant at steady-state when selection is strongest at the between-group level ($\lambda \to \infty$). In the latter case, the effect of within-group selection can be observed for any relative strength of selection at the two levels, as the mean and modal groups at steady-state feature fewer cooperators than is optimal for between-group competition. This effect is still felt even in the limit as between-group competition becomes much stronger than within-group competition (i.e. $w_G \gg w_I$ or $\lambda \to \infty$). A similar phenomenon is present in the Hawk-Dove game, although the effects of within-group selection are realized through groups achieving a fraction of cooperators near the mixed-strategy equilibrium $x^{eq} = - \tfrac{\beta}{\alpha}$ for the replicator dynamic in a well-mixed population. 

\myindent This discrepancy between the behavior of this multilevel selection model when group reproductive success is optimized at an intermediate or extremal cooperator frequency could shed light on questions related to the emergence of higher levels of biological organization.  For situations in which one type of individual is purely dominant at a lower level of organization (type-$I$ in Luo's terminology) and another type is purely dominant at a higher level of organization (type-$G$), our analysis has shown that the long-time distribution of the two types is determined by the relative strength of the two levels.  Further, strong enough selection at the between-group level is sufficient to achieve abundant cooperation at steady-state when all-cooperator groups have the highest average payoff. However, when mixed groups are most fit at the between-group level, it turns out that this form of deterministic multilevel selection is not sufficient to achieve the emergence of optimal group average payoff at steady-state, and the multilevel system can never eliminate the ghost of lower-level selection. As such, problems with intermediate group payoff optima are not as simple to solve. In the extreme case of group payoff function $G(x) = x(1-x)$ for the Prisoner's Dilemma, no cooperation could ever succeed for any relative selection strength at the two levels, and so alternative mechanisms are necessary to achieve emergence of cooperation in such scenarios. 

\myindent The threshold criteria in Propositions \ref{prop:steadystatespecialpd} and \ref{prop:steadystatespecialhd} for the level of relative between-group selection strength $\lambda$ needed to exhibit long-time cooperation highlight the key importance of all-cooperator groups in sustaining cooperative outcomes for the whole metapopulation. These criteria also provide some analytical support for the numerical intuition provided by Markvoort et al \citep{markvoort2014computer} for the necessity of initial all-cooperator groups to promote evolutionary coexistence of cooperators and defectors. The success of all-cooperator groups also plays an important role in Traulsen-Nowak stochastic model of multilevel selection \citep{traulsen2008analytical}, in which there is a separation of time scales such that selection of cooperators or defectors within groups acts much faster than selection between groups. That this effect appears even in the continuum limit and in cases where either within-group selection or between-group selection most favors a mix of cooperators and defectors indicates that the success of all-cooperator groups is truly necessary for the viability of the whole population.

\myindent In this paper, we have primarily focused on one family of Prisoner's Dilemmas and one family of Hawk Dove games. For these families, the solvability of the within-group replicator dynamics allowed us to use the method of characteristics to fully characterize the long-time behavior of the corresponding special cases of Equation \ref{eq:generalPDE}, and helped to provide insight into the properties of the steady-state behavior as $\lambda$ ranges from $0$ to $\infty$. However, it is interesting to wonder how general these results are, particularly as to whether one can similarly characterize the asymptotic behavior of Equation \ref{eq:generalPDE} without solvable characteristic curves, and whether there is a similar discrepancy in qualitative behavior when group payoff $G(x)$ is maximized by all-cooperator groups and when groups are best off with a mix of cooperators and defectors. In addition, one could explore the breadth of biological phenomena that have as a natural description a competition between the effects at two levels of selection.

\myindent Given the gap in cooperation displayed by the mean and modal groups at steady-state relative to intermediate group-level optima even for strong between-group competition, another problem which arises is how even higher levels of selection could emerge if the effects of the lowest level of selection can never be forgotten. Is there a way to overcome this ghost of lower level selection? Perhaps the conflict between the two initial levels of selection must be resolved through an alternative mechanism which eliminates the misalignment of incentives at the two levels. Pruitt et al have observed that individual-level regulation of birth rates of cheaters and cooperators can stabilize optimal trait fractions in nest-structured spider populations \citep{pruitt2014site,pruitt2017intense}, while Haig has suggested mechanisms for intracellular and intercellular selection which allow mitochondrial DNA to prevent the uncontrolled replication that diminishes cellular function \citep{haig2016intracellular}. Aktipis et al characterize five mechanisms by which organisms can maintain healthy ratios of somatic cells, and demonstrate that the failure of such mechanisms and the dominance of ``defector'' cells constitute the onset of cancer \citep{aktipis2015cancer}. With these examples in mind, a natural future direction for our multilevel selection model is to explore the potential mechanisms by which groups can overcome this ghost of lower level selection and potentially evolve towards higher levels of biological complexity.

\renewcommand{\abstractname}{Acknowledgments}
\begin{abstract} 
 This research was supported by NSF grants DMS-1514606 and GEO-1211972. I would like to thank Carl Veller for initial discussions and advice on the problem of multilevel selection in evolutionary games. I am grateful to Carl Veller, Simon Levin, Joshua Plotkin, and Chai Molina for helpful comments on the manuscript and to Peter Constantin, Robin Pemantle, Fernando Rossine, Dylan Morris, George Constable, and Gergely Boza for helpful discussions. %

\end{abstract}

\bibliography{multilevelselection}
\bibliographystyle{ieeetr}


\appendix

\section{Integrals for Prisoner's Dilemma}

In this section, we will use the solutions to Equation \ref{eq:specialPDrep}, the replicator dynamics for the special PD from Section \ref{sec:PDsolvable}, in order to express the solution of Equation \ref{eq:specialPDPDE} along characteristics in terms of the integrals $\int_0^t x(x_0,s) ds$ and $\int_0^t x(x_0,s)^2 ds$ . We use the formulas for both the forward solution $x(t,x_0) = \phi_t(x_0)$ given by Equation \ref{eq:specialPDchar} and the backward solution $\phi^{-1}_t(x) = x_0(t,x)$ given by Equation \ref{eq:specialPDbackwardschar}.

The easier integral for us to compute is \begin{align*} \ds\int_0^t x(s)^2 ds &= \ds\int_0^t \frac{ds}{e^{-2s} \left( \frac{1}{x_0^2} - 1 \right) + 1} = \frac{1}{2} \ds\int_1^{e^{2t}} \frac{du}{u\left(u \left(\frac{1}{x_0^2} - 1\right) + 1\right)} = \frac{1}{2}\ds\int_1^{e^{2t}}  \left( \frac{1}{u} + \frac{\left(\frac{1}{x_0^2} - 1\right)}{u \left(\frac{1}{x_0^2} - 1\right)+ 1}\right) du \\ &= \frac{1}{2} \left( \log(u) + \log\left(\left(\frac{1}{x_0^2} - 1\right) u + 1 \right) \right) \bigg|_1^{e^{2t}} = t + \frac{1}{2} \log\left( (1 - x_0^2) e^{2t} + x_0^2 \right) \end{align*}
Then, plugging in for $x_0^2 = \left[\left( \frac{1}{x^2} - 1 \right)e^{-2t} + 1\right]^{-1}$ gives us that \begin{equation} \label{eq:xssquaredintegral} \ds\int_0^t x(s)^2 ds = t + \frac{1}{2} \log \left( \left( 1 - \left( \frac{1}{x^2} - 1 \right) \right) e^{2t} + \left( \frac{1}{x^2} - 1 \right) \right)  = t - \frac{1}{2} \log \left(e^{-2t} \left( 1 - x^2 \right) + x^2 \right) \end{equation}
The more involved integral is $$\ds\int_0^t x(s) ds  =  \ds\int_0^t \frac{ds}{\sqrt{e^{-2s} \left( \frac{1}{x_0^2} - 1 \right) + 1}} = \frac{1}{2} \ds\int_{\frac{1}{x_0^2}}^{e^{2t}\left(\frac{1}{x_0^2} - 1 \right) + 1} \frac{du}{(u-1) u^{1/2}}$$ where $u = e^{2s} \left( \frac{1}{x_0^2} - 1 \right) + 1$. Partial fraction decomposition gives us \begin{align*}  \ds\int_0^t x(s) ds &= \frac{1}{4} \ds\int_{\frac{1}{x_0^2}}^{e^{2t}\left(\frac{1}{x_0^2} - 1 \right) + 1} du \left(\frac{2}{u^{1/2}} + \frac{1}{1 - u^{1/2}} - \frac{1}{1 + u^{1/2}} \right) \\ &= \frac{1}{2} \left( - \log \left(1 + u^{1/2} \right) + \log \left(|1 - u^{1/2} |\right) \right)\bigg|_{\frac{1}{x_0^2}}^{e^{2t}\left(\frac{1}{x_0^2} - 1 \right) + 1} \\ &=   \frac{1}{2} \left( - \log \left(1 - \sqrt{e^{2t} \left( \frac{1}{x_0^2} - 1\right) + 1} \right) + \log\left(1 + \frac{1}{x_0} \right)\right) \\ &+ \frac{1}{2} \left( \log\left( - 1 + \sqrt{e^{2t} \left( \frac{1}{x_0^2} - 1\right) + 1}  \right) - \log\left(\frac{1}{x_0} - 1 \right) \right) \end{align*}

Using $x_0^2 = \frac{1}{e^{-2t}\left(\frac{1}{x^2} - 1\right) + 1}$, we see that $1 - \sqrt{e^{2t} \left( \frac{1}{x_0^2} - 1\right) + 1} = 1 + \frac{1}{x}$ and $-1  + \sqrt{e^{2t} \left( \frac{1}{x_0^2} - 1\right) + 1} = 1 - \frac{1}{x}$, and we can deduce that $$\ds\int_0^t x(s) ds = \frac{1}{2} \left(\log\left(\frac{1-x}{1+x} \right)  + \log \left( \frac{1 + x_0}{1 - x_0} \right) \right) =  \frac{1}{2} \left(\log\left(\frac{1-x}{1+x} \right)  + \log \left( \frac{(1 + x_0)^2}{1 - x_0^2} \right) \right)  $$ Again using the definition of $x_0$ we see that 

\begin{align*} \frac{(1+x_0)^2}{1 - x_0^2} &= \left[\frac{\left(\sqrt{e^{-2t} \left(\frac{1}{x^2} - 1 \right) + 1} + 1 \right)^2 }{e^{-2t} \left(\frac{1}{x^2} - 1 \right) + 1} \right] \bigg/  \left[1 - \frac{1}{e^{-2t} \left(\frac{1}{x^2} - 1 \right) + 1} \right] =
\frac{\left(\sqrt{e^{-2t} \left(\frac{1}{x^2} - 1 \right) + 1} + 1 \right)^2}{e^{-2t} \left(\frac{1}{x^2} - 1\right)} \\ &=
e^{2t} \left(e^{-2t} + \frac{2 x^2 + 2x^2 \sqrt{e^{-2t} \left(\frac{1}{x^2} - 1 \right) + 1}}{(1-x)(1+x)} \right) \end{align*} %
which allows us to deduce that \begin{equation} \label{eq:xsintegral} \ds\int_0^t x(s) ds = t + \frac{1}{2} \log \left(\frac{1+x}{1-x} \right) + \frac{1}{2} \log \left( \left(e^{-2t} + \frac{2 x^2 + 2x^2 \sqrt{e^{-2t} \left(\frac{1}{x^2} - 1 \right) + 1}}{(1-x)(1+x)}  \right)  \right) \end{equation}

\section{Integrals for Hawk-Dove}

In this section, we will use the solutions to Equation \ref{eq:specialHDrep}, the replicator dynamics for the special PD from Section \ref{sec:HDsolvable}, in order to express the solution of Equation \ref{eq:hawkdovePDE} along characteristics in terms of the integrals $\int_0^t x(x_0,s) ds$ and $\int_0^t x(x_0,s)^2 ds$ . We use the formulas for both the forward solution $x(t,x_0) = \phi_t(x_0)$ given by Equation \ref{eq:hdcharacteristics} and the backward solution $\phi^{-1}_t(x) = x_0(t,x)$ given by Equation \ref{eq:hdcharbackwards}.

\subsection{$x(t),x_0 > \frac{1}{2}$}
Here we compute $\ds\int_0^t x(s) ds$ for $x(t),x_0 > \tfrac{1}{2}$. Making the substitution $p = 2x_0 - 1 > 0$, we have \begin{align*} \ds\int_0^t x(s) ds &= \ds\int_0^t \left[ \frac{1}{2} \left(1 + \frac{p}{\sqrt{(p^2 + (1 - p^2) e^s}} \right) \right] ds \\ &= \frac{t}{2} + \frac{p}{2} \ds\int_0^t \frac{ds}{\sqrt{p^2 + (1-p^2) e^s}} \\ &= \frac{t}{2} +  \frac{p}{2} \ds\int_1^{p^2 + (1-p^2)e^t} \frac{du}{(u - p^2) \sqrt{u}} \; \: \: \textnormal{where  } u = p^2 + (1-p^2)e^s \\ &= \frac{t}{2} + \frac{1}{2p} \ds\int_1^{p^2 + (1-p^2)e^t} \left(\frac{-1}{u^{1/2}} + \frac{1}{2} \frac{1}{u^{1/2} - p} + \frac{1}{2} \frac{1}{u^{1/2} + p} \right) du \\ &= \frac{t}{2} + \frac{1}{2} \left[ \log\left(|u^{1/2} - p|\right) - \log \left( |u^{1/2} + p| \right) \right] \bigg|_1^{p^2 + (1-p^2) e^t} \\ &= \frac{t}{2} + \frac{1}{2} \left(\log \left(\frac{1+p}{1-p}\right) + \log \left(\frac{\sqrt{p^2 + (1-p^2)e^t} - p}{p + \sqrt{p^2 + (1-p^2)e^t}} \right) \right) \end{align*}

Now we can write $p = 2x_0 - 1$ as a function of $x$ to write our integral in terms of $x$. We have that $$p^2 = (2x_0 - 1)^2 = \frac{(2x-1)^2}{(2x-1)^2 + (1 - (2x-1)^2)e^{-t}}  \Longrightarrow p^2 + (1-p^2) e^t = \frac{1}{(2x-1)^2 + (1 - (2x-1)^2)e^{-t}}$$ and we can compute that $$\frac{\sqrt{p^2 + (1-p^2) e^t} - p}{\sqrt{p^2 + (1-p^2)e^t} + p} = \left(\frac{1 - (2x -1)}{\sqrt{(2x-1)^2 + (1 - (2x-1)^2)e^{-t}}} \right)  \left(\frac{\sqrt{(2x-1)^2 + (1 - (2x-1)^2)e^{-t}}}{1 + (2x - 1)} \right)= \frac{1-x}{x} $$

We can also find that \begin{align*} \frac{1+p}{1-p} = \frac{(1+p)^2}{1 - p^2} &=  \left( 1 + \frac{2x-1}{\sqrt{(2x-1)^2 + (1 - (2x-1)^2)e^{-t}}} \right)^2  \bigg/ \left(\frac{(1 - (2x-1)^2) e^{-t}}{(2x-1)^2 + (1 - (2x-1)^2)e^{-t}} \right) \\ &= e^t \left[ \frac{\left(\sqrt{(2x-1)^2 + (1 - (2x-1)^2)e^{-t}} + (2x -1)\right)^2}{4 x(1-x)} \right]\end{align*} and we can use the two equations above to find that \begin{equation} \label{hdxequation} \ds\int_0^t x(s) ds =  t + \log \left( \frac{\sqrt{(2x-1)^2 + (1 - (2x-1)^2)e^{-t}} + (2x -1)}{2 x} \right) \end{equation}

Next we would like to integrate $x(s)^2$, which we can see that \begin{align*} \ds\int_0^t x(s)^2 ds &= \frac{1}{4} \ds\int_0^t \left( 1 + \frac{p}{\sqrt{p^2 + (1-p^2)e^{s}}}  \right)^2 ds  \\&= \frac{t}{4} + \frac{p}{2} \ds\int_0^t  \frac{ds}{\sqrt{p^2 + (1-p^2)e^{s}}} + \frac{p^2}{4} \ds\int_0^t \frac{ds}{p^2 + (1-p^2) e^s} \\ &= \frac{t}{4} + \left(t + \log \left( \frac{\sqrt{(2x-1)^2 + (1 - (2x-1)^2)e^{-t}} + (2x -1)}{2 x} \right)  \right) + \frac{p^2}{4} \ds\int_1^{p^2 + (1-p^2)e^t} \frac{1}{u(u - p^2)} du \end{align*}

For the last term, we see that \begin{align*} \frac{p^2}{4} \ds\int_1^{p^2 + (1-p^2)e^t} \frac{1}{u(u - p^2)} du  &= \frac{1}{4} \ds\int_1^{p^2 + (1-p^2)e^t} \left(\frac{1}{u - p^2} - \frac{1}{u} \right) du %
  = \frac{1}{4} \left[t + \log \left( (2x-1)^2 + (1 - (2x-1)^2)e^{-t} \right) \right] \end{align*} Putting all of this together, we find that 

\begin{equation} \label{hdx2equation} \ds\int_0^t x(s)^2 ds = t + \log \left( \frac{\sqrt{(2x-1)^2 + (1 - (2x-1)^2)e^{-t}} + (2x -1)}{2 x \left( (2x-1)^2 + (1 - (2x-1)^2)e^{-t}  \right) ^{-1/4}} \right) \end{equation}

\subsection{$x(s),x_0 < \frac{1}{2}$}

For $x(s) < \frac{1}{2}$ (and correspondingly $x_0 < \frac{1}{2}$), we have that $$x(s) = \frac{1}{2} \left(1 - \frac{1 - 2x_0}{\sqrt{(1 - 2x_0)^2 + (1 - (1-2x_0)^2}e^{-t}} \right)  = \frac{1}{2} \left( 1 - \frac{q}{\sqrt{q^2 + (1-q^2)e^{-t}}}\right)$$ where $q = 1 - 2x_0 > 0$ when $x_0 < \frac{1}{2}$. Then we can the integrals computed above to see that \begin{align*} \ds\int_0^t x(s) ds &= \frac{1}{2} \ds\int_0^t  \left( 1 - \frac{q}{\sqrt{q^2 + (1-q^2)e^{-s}}}\right)  ds  \\ &= \frac{t}{2} - \frac{p}{2} \ds\int_0^t \frac{1}{\sqrt{q^2 + (1-q^2)e^{-s}}}  ds \\ &= \frac{t}{2} -  \left( \frac{t}{2} + \log \left( \frac{\sqrt{(1-2x)^2 + (1 - (1-2x)^2)e^{-t}} + (1-2x)}{2 x} \right) \right) \\ &= - \log \left( \frac{\sqrt{(1-2x)^2 + (1 - (1-2x)^2)e^{-t}} + (1-2x )}{2 x} \right) \end{align*}

Similarly, we can compute \begin{align*} \ds\int_0^t x(s)^2 ds &= \frac{1}{4} \ds\int_0^t \left( 1 - \frac{q}{\sqrt{q^2 + (1-q^2)e^{-s}}}\right)^2 ds \\ &= \frac{t}{4} - \frac{q}{2} \ds\int_0^t  \frac{1}{\sqrt{q^2 + (1-q^2)e^{-s}}}  ds + \frac{q^2}{4} \ds\int_0^t \frac{ds}{q^2 + (1-q^2)e^{-s}} \\ &= \frac{t}{4} - \left( \frac{t}{2} + \log \left( \frac{\sqrt{(1-2x)^2 + (1 - (1-2x)^2)e^{-t}} + (1-2x)}{2 x} \right)\right) \\ &+ \frac{t}{4} + \frac{1}{4}  \log \left( (1-2x)^2 + (1 - (1-2x)^2)e^{-t} \right) \\ &=  \frac{1}{4} \log \left( (1-2x)^2 + (1 - (1-2x)^2)e^{-t} \right) - \log \left( \frac{\sqrt{(1-2x)^2 + (1 - (1-2x)^2)e^{-t}} + (1-2x)}{2 x}  \right) \end{align*} 

In particular, we observe that our expressions for $\int_0^t x(s) ds$ and $\int_0^t x(s)^2 ds$ do not have any linear factors of $t$ when $x(s) < \frac{1}{2}$, which are notably present in our expressions when $x(s) > \frac{1}{2}$. We can use this to show how the probability of group compositions below the Hawk-Dove equilibrium vanishes as $t \to \infty$.

\section{Preservation of Hölder exponent near $x=1$} \label{sec:Holderpreserved}

In this section, we study how the Hölder exponent of the population state $f(t,x)$ evolves in time. Using analytically-solvable special cases of Equation \ref{eq:generalPDE} with illustrative choices of initial data, we demonstrate examples in which the Hölder exponent of the distribution of cooperators near $x=1$ is preserved by the dynamics of our multilevel system. A proof that the Hölder exponent near $x=1$ is preserved for solutions to more general cases to Equation \ref{eq:generalPDE} will be presented in future work.

\subsection{Frequency Independent (Luo-Mattingly) Case}

For the frequency independent model, described by Equation \ref{eq:luomattingly}, we can use the reasoning of Example \ref{ex:exactsolution} and the exact solution for uniform initial data given by Luo and Mattingly \citep{luo2017scaling} to calculate exact solutions for our choice of initial data. For density-valued initial data of the form $f(0,x) = \theta (1-x)^{\theta - 1}$, we see that Equation \ref{eq:luomattingly} has density-valued solutions given by 
\begin{equation} \label{eq:lmexactsolution} f_{\theta}(t,x) = \frac{\theta (\lambda - 1) \left(1 - e^{-t}\right)}{1 - e^{- (\lambda - 1) t}} \left[ \left( 1 - x \right)e^{-t} + x \right]^{\lambda - \theta - 1} \left( 1 -x \right)^{\theta - 1} e^{- (\theta - 1) t}  \end{equation}

Then we use L'H\^opital's rule to see that 
\[ \ds\lim_{x \to 0} \frac{\ds\int_{1-x}^1 f_{\theta}(t,x)}{x^{\Theta}} = \ds\lim_{x \to 0} \frac{f_{\theta}(t,1-x)}{\Theta x^{\Theta - 1}} = \frac{\theta}{\Theta} \left[\ds\lim_{x \to 0} x^{\theta - \Theta}\right] \left[ \ds\lim_{x \to 0} \left( \frac{(\lambda - 1) e^{-(\theta -1) t} \left(1 - e^{-t}\right)}{1 - e^{- (\lambda - 1) t}} \left[x e^{-t} + (1-x) \right]^{\lambda - \theta - 1} \right) \right] \] 
Further noting that \[  \ds\lim_{x \to 0} \left( \frac{(\lambda - 1) e^{-(\theta -1) t} \left(1 - e^{-t}\right)}{1 - e^{- (\lambda - 1) t}} \left[x e^{-t} + (1-x) \right]^{\lambda - \theta - 1} \right) = \frac{(\lambda - 1) e^{-(\theta -1) t} \left(1 - e^{-t}\right)}{1 - e^{- (\lambda - 1) t}} := K > 0  \]
Then we can see that 
 \begin{displaymath}
    \ds\lim_{x \to 0} \frac{\ds\int_{1-x}^1 f_{\theta}(t,x)}{x^{\Theta}} = \left\{
     \begin{array}{ll}
       0 & : \Theta < \theta \\
       K & : \Theta = \theta  \\
     \infty &: \Theta > \theta \end{array}
   \right.
\end{displaymath} 
so we can deduce that the Hölder exponent of the measure $f_{\theta}(t,x) \dx$ with corresponding density $f_{\theta}(t,x)$ is equal to $\theta$ for all times $t \geq 0$. Thus, for the family of initial conditions $f_0(x) = \theta(1-x)^{\theta -1}$ (with initial Hölder exponent $\theta$), the Hölder exponent for the population near $x=1$ is preserved in time.

\subsection{Prisoner's Dilemma with $\beta = -1$, $\alpha = -1$}

We now consider the special family of Prisoner's Dilemmas whose behavior is the subject of Section \ref{sec:PDsolvable}. Recalling Equation \ref{eq:PDspecialsolution}, we have that 
\[ f(t,x) = \theta Z_f^{-1}e^{\left[\lambda  \left(\gamma - 1 \right) - 2 \theta\right] t} g_t(x) \left(  \left(1 - x^2 \right) + e^{2t} \left(x^2 - x \sqrt{e^{-2t}\left(1 - x^2\right) + x^2} \right) \right)^{\theta - 1} \left(e^{-2t} (1-x^2) + x^2 \right)^{1 - \theta}  \]
We can then compute that 
\begin{align}\ds\lim_{x \to 0} \left[  \frac{\int_{1-x}^1 f(t,y)\dy}{x^{\Theta}} \right] &= \ds\lim_{x \to 0} \frac{f(t,1-x)}{\Theta x^{\Theta - 1}} \nonumber \\ \label{eq:Holderdensity} &= \frac{\theta}{\Theta} Z_f^{-1} e^{\left[\lambda  \left(\gamma - 1 \right) - 2 \theta\right] t} \ds\lim_{x \to 0} \left[  \frac{1}{ x^{\Theta -1}}  g_t(1-x)  \left(e^{-2t} (2x - x^2) + (1-x)^2 \right)^{1 - \theta} \left(A(x,t)\right)^{\theta - 1} \right]\end{align}
where $A(x,t) :=  \left(2x - x^2 \right) + e^{2t} \left((1-x)^2 - (1-x) \sqrt{e^{-2t}\left(2x - x^2\right) + (1-x)^2} \right) $.
To simplify the expression in Equation \ref{eq:Holderdensity}, we first note that $\lim_{x \to 0} \left[e^{-2t} (2x - x^2) + (1-x)^2 \right]^{1 - \theta}  =  1$. Further, recalling Equation \ref{eq:gtx}, we have that 
\begin{equation}   \label{eq:gt1-x}  g_t(1-x) =%
 \left( \frac{x e^{-2t}}{2-x} + \frac{2(1-x)^2 + 2 (1-x) \sqrt{e^{-2t} \left(2x - x^2 \right) + (1-x)^2 }}{(2-x)^2} \right)^{\frac{\lambda \gamma}{2}} \left( e^{-2t} (2x - x^2) + (1-x)^2 \right)^{-\frac{\lambda}{2}} \end{equation}
and we see that $\lim_{x \to 0} g_t(1-x) = 1$.
This allows us to say that \begin{align*}\ds\lim_{x \to 0} \left[  \frac{\int_{1-x}^1 f(t,y)\dy}{x^{\Theta}} \right] &= \frac{\theta}{\Theta} Z_f^{-1} e^{\left[\lambda  \left(\gamma - 1 \right) - 2 \theta\right] t} \ds\lim_{x \to 0} \left[ \frac{\left(A(x,t)\right)^{\theta -1}}{x^{\Theta -1}} \right] \\ &= \frac{\theta}{\Theta} Z_f^{-1} e^{\left[\lambda  \left(\gamma - 1 \right) - 2 \theta\right] t}  \left( \ds\lim_{x \to 0} \left[ \frac{A(x,t)}{x}\right]\right)^{\theta -1} \ds\lim_{x \to 0} \left[ x^{\theta - \Theta} \right] \end{align*} 
Then we see that 
\begin{align*} \ds\lim_{x \to 0} \frac{A(x)}{x} &= \ds\lim_{x \to 0} \left[ \frac{2x - x^2}{x} + \frac{e^{2t}}{x} \left((1-x)^2 - (1-x) \sqrt{e^{-2t} (2x - x^2) + (1-x)^2} \right) \right] \\  &= 2 + e^{2t}\ds\lim_{x \to 0} \left[ - 2 (1-x) + \sqrt{e^{-2t}(2x-x^2) + (1-x)^2} - (1-x)^2 \frac{ e^{-2t} - 1 }{ \sqrt{e^{-2t}(2x-x^2) + (1-x)^2}} \right]  =  1 \end{align*}
and therefore we see that 
\begin{displaymath}
    \ds\lim_{x \to 0} \frac{\ds\int_{1-x}^1 f(t,y) \dy}{x^{\Theta}} = \left\{
     \begin{array}{lr}
       0 & : \Theta < \theta\\
     Z_f^{-1} e^{\left[\lambda  \left(\gamma - 1 \right) - 2 \theta\right] t}  & : \Theta = \theta  \\
     \infty &: \Theta > \theta \end{array}
   \right.
\end{displaymath} 
and so we see that the Hölder exponent of $f(t,x)$ near $x=1$ agrees with the Hölder exponent of the initial distribution when initial distributions are chosen to have the form $\theta (1-x)^{\theta - 1}$. 
%


\end{document}